\documentclass[a4paper,UKenglish,cleveref, autoref, thm-restate]{lipics-v2021}
\hideLIPIcs  %uncomment to remove references to LIPIcs series (logo, DOI, ...), e.g. when preparing a pre-final version to be uploaded to arXiv or another public repository

\title{On Knowledge Compilation Languages for QBFs\footnote{A version of this paper is accepted to be published in proceedings of the 46th IARCS Annual Conference on Foundations of Software Technology and Theoretical Computer Science (FSTTCS), 2026.}} %TODO Please add

\author{Sravanthi Chede}{The Institute of Mathematical Sciences (a CI of Homi Bhabha National Institute), Chennai, India \and \url{https://sravanthi-ch.github.io/webpage/} }{sravanthic@imsc.res.in}{https://orcid.org/0000-0001-7170-6156}{Supported by the ANRF J C Bose fellowship no. JCB/2023/000006.}%TODO mandatory, please use full name; only 1 author per \author macro; first two parameters are mandatory, other parameters can be empty. Please provide at least the name of the affiliation and the country. The full address is optional. Use additional curly braces to indicate the correct name splitting when the last name consists of multiple name parts.

\author{Vaibhav Krishan}{The Institute of Mathematical Sciences (a CI of Homi Bhabha National Institute), Chennai, India\and \url{https://vaibhkrishan.github.io/} }{vaibhkrishan@gmail.com}{https://orcid.org/0009-0000-0335-1963}{}

\author{Anil Shukla}{Indian Institute of Technology Ropar, Rupnagar, India \and \url{https://anilshukl.github.io/website/} }{anilshukla@iitrpr.ac.in}{https://orcid.org/0009-0009-9051-4374}{}

\authorrunning{S. Chede, V. Krishan and A. Shukla} %TODO mandatory. First: Use abbreviated first/middle names. Second (only in severe cases): Use first author plus 'et al.'

\Copyright{Sravanthi Chede, Vaibhav Krishan and Anil Shukla} %TODO mandatory, please use full first names. LIPIcs license is "CC-BY";  http://creativecommons.org/licenses/by/3.0/

 %TODO mandatory: Please choose ACM 2012 classifications from https://dl.acm.org/ccs/ccs_flat.cfm 

 %TODO mandatory; please add comma-separated list of keywords

\category{} %optional, e.g. invited paper

\relatedversion{} %optional, e.g. full version hosted on arXiv, HAL, or other respository/website
\supplementdetails[ linktext={Github:vaibhkrishan/TQTSolver},
  subcategory={TQTSolver Repository}
]{Software}{https://github.com/vaibhkrishan/TQTSolver.git}

\supplementdetails[
linktext={doi:10.5281/zenodo.23276976},
  subcategory={TQTSolver and Benchmarks}
]{Software and Dataset}{https://doi.org/10.5281/zenodo.23276976}

\acknowledgements{We sincerely thank Swagato Sanyal for valuable discussions and support throughout this work. We are thankful for interesting discussions with Olaf Beyersdorff, Meena Mahajan, Lea Kasche and Satya Amar. We also thank the anonymous reviewers for useful suggestions.
}%optional

\nolinenumbers %uncomment to disable line numbering

\EventEditors{Raghavan Komondoor, Ramprasad Saptharishi, and James Worrell}
\EventNoEds{3}
\EventLongTitle{46th IARCS Annual Conference on Foundations of Software Technology and Theoretical Computer Science (FSTTCS 2026)}
\EventShortTitle{FSTTCS 2026}
\EventAcronym{FSTTCS}
\EventYear{2026}
\EventDate{December 16--18, 2026}
\EventLocation{IIT Delhi, New Delhi, India}
\EventLogo{}
\SeriesVolume{400}
\ArticleNo{32}
\usepackage{graphicx}
\usepackage{lipsum}
\usepackage{amsmath}
\usepackage{makecell,array}
\usepackage{graphicx}

\usepackage{float}
\usepackage{gensymb}
\usepackage{adjustbox}
\usepackage{sectsty}
\usepackage{graphicx}
\usepackage{bussproofs}
\usepackage{bbm}
\usepackage{float}
\usepackage{xspace}
\usepackage{xifthen}
\usepackage{turnstile}
\usepackage{bm}
\usepackage{diagbox}
\usepackage{algorithm}
\usepackage{algpseudocode}
\usepackage{graphicx}
\usepackage[labelfont=bf]{caption}
\usepackage[labelfont=bf]{subcaption}
\usepackage{amsmath}
\usepackage{amssymb}

\newcommand{\Red}{\Math{\boldsymbol{\forall}\kern .04ex\ProofSystemsFont{red}}}

\newcommand{\red}{\Math{\,\raisebox{.2ex}{$\scriptstyle+$}\,\Red}}
\newcommand{\Min}{{\mathsf{Min}}}

\newlength{\breite}
\newcommand{\Math}[1]{\ensuremath{#1}}
\newcommand{\ComplexityClassFont}[1]{\mathsf{#1}}

\newcommand{\ProofSystemsFont}[1]{\mathsf{#1}}
\newcommand{\DefineProofSystem}[2]{\expandafter\def\csname#1\endcsname{\Math{\ComplexityClassFont{#2}}\xspace}}
\newcommand{\DefineProofSystemNoSpace}[2]
{\expandafter\def\csname#1\endcsname{\Math{\ComplexityClassFont{#2}}}}
\newcommand{\Frege}[1][]{\Math{\ifthenelse{\isempty{#1}}
{\ProofSystemsFont{Frege}}{#1\text{-}\ProofSystemsFont{Frege}}}}
\DefineProofSystem{eFrege}{e\Frege[]\xspace}
\newcommand{\eFregeRed}[1][]{\eFrege\!\!\red}
\DefineProofSystemNoSpace{DRAT}{\ProofSystemsFont{DRAT}}
\DefineProofSystem{MCPS}{\ProofSystemsFont{CLIP}}
\DefineProofSystemNoSpace{clip}{\ProofSystemsFont{CLIP}}
\DefineProofSystem{micep}{\ProofSystemsFont{MICE'}}
\DefineProofSystem{cpog}{\ProofSystemsFont{CPOG}}
\DefineProofSystem{kcps}{\ProofSystemsFont{KCPS(\#SAT)}}
\DefineProofSystem{fxyz}{\ProofSystemsFont{Truncated~ Quantified~ Tree}}
\DefineProofSystem{xyz}{\ProofSystemsFont{TQT}}
\DefineProofSystem{xyzs}{\ProofSystemsFont{TQT}s}
\DefineProofSystem{fcxyz}{\ProofSystemsFont{Certified~ Truncated~ Quantified~ Tree}}
\DefineProofSystem{cxyz}{\ProofSystemsFont{cert\text{-}TQT}}
\DefineProofSystem{dxyz}{\ProofSystemsFont{TQD}}
\DefineProofSystem{fdxyz}{\ProofSystemsFont{Truncated~ Quantified~ DAG}}
\DefineProofSystem{cdxyz}{\ProofSystemsFont{cert\text{-}TQD}}
\DefineProofSystem{fcdxyz}{\ProofSystemsFont{Certified~ Truncated~ Quantified~ DAG}}

\DefineProofSystem{mice}{\ProofSystemsFont{MICE}}
\DefineProofSystem{qmice}{\ProofSystemsFont{Q\text{-}MICE}}
\DefineProofSystem{emice}{\Math{\forall}\ProofSystemsFont{Exp\text{+}MICE}}
\DefineProofSystem{exp}{\Math{\forall}\ProofSystemsFont{Exp\text{+}}\mathcal{P}}
\DefineProofSystem{qbfsolver}{\ProofSystemsFont{TQTSolver}}

\newcommand{\abs}[1]{\ensuremath{\lvert #1 \rvert}}

\newcommand{\vars}{\mathrm{\textit{vars}}}

\usepackage[dvipsnames]{xcolor}

\usepackage{tikz}
\usetikzlibrary{shapes,automata, matrix, arrows.meta, positioning}
\usetikzlibrary{decorations.pathreplacing}

\usetikzlibrary{fadings}
\tikzstyle{uedge}=[draw=blue!50!red]
\tikzstyle{fedge}=[draw=blue]
\tikzstyle{iedge}=[draw=red]
\tikzstyle{redge}=[draw=green!50!black]
\tikzstyle{rnode}=[draw,inner sep=2pt,color=black]
\tikzstyle{tnode}=[circle,minimum width=3pt,fill,inner sep=0pt]
\tikzstyle{dotnode}=[circle,minimum width=2pt,fill,inner sep=0pt]
\tikzstyle{labn}=[font=\sffamily,circle,fill=white,inner sep=1pt,draw=black]
\tikzstyle{legn}=[font=\scriptsize]
\tikzstyle{reln}=[circle,fill=white,inner sep=.4pt,draw=black]
\tikzstyle{oreln}=[circle,fill=white,inner sep=.4pt,draw=black!50,solid]
\tikzstyle{exist}=[thick, circle,fill=white,inner sep=3pt,draw=black,solid,draw]
\tikzstyle{univ}=[thick, rectangle,fill=white,inner sep=5pt,draw=black,solid,draw]
\tikzstyle{oree}=[thick,draw=black!50,densely dashed]
\tikzstyle{ree}=[thick,draw=black]
\tikzstyle{calcn}=[rectangle%
,rounded corners=1mm
,draw=black,
minimum height=.5cm,minimum width=.5cm,inner sep=5pt%
]
\tikzset{cross/.style={cross out, draw=black, fill=none, minimum size=.4cm, inner sep=0pt, outer sep=0pt}, cross/.default={1pt}}

\newtheoremstyle{TheoremNum}
    {\topsep}{\topsep}              %%% space between body and thm
    {\itshape}                      %%% Thm body font
    {}                              %%% Indent amount (empty = no indent)
    {\bfseries}                     %%% Thm head font
    {}                             %%% Punctuation after thm head
    { }                             %%% Space after thm head
    {\thmname{#1}\thmnote{ \textcolor{darkblue}{\bfseries #3}}}%%% Thm head spec
\theoremstyle{TheoremNum}

\newtheoremstyle{TheoremNum2}
    {\topsep}{\topsep}              %%% space between body and thm
    {\normalfont}                      %%% Thm body font
    {}                              %%% Indent amount (empty = no indent)
    {\itshape}                     %%% Thm head font
    {.}                             %%% Punctuation after thm head
    { }                             %%% Space after thm head
    {\thmname{#1}\thmnote{ \textcolor{darkblue}{\bfseries #3}}}%%% Thm head spec
\theoremstyle{TheoremNum2}

\begin{document}

\maketitle
%TODO mandatory: add short abstract of the document
\begin{abstract}
In this work, we design knowledge compilation languages for quantified Boolean formulas (QBFs), with a particular focus on model counting. We introduce various languages in increasing order of their power of representation, starting from a simple tree-based structure, and extending it to DAGs for a more efficient representation. We study the power and limitations of each language. We also show that QBF queries such as model counting and MaxQBF are easy to handle in all these languages.

These languages naturally lead to various static proof systems for the \#QBF problem, the model-counting version of the QBF problem.
We compare these systems with one another and with a recently introduced line-based \#QBF-system called Q-MICE (Chede, Chew, Krishan and Shukla, SAT `26). We show that a minor extension of the Q-MICE system efficiently simulates all these static proof systems, and is in fact exponentially stronger than each of them.

We also initiate the implementation of a solver based on these languages.
The solver compiles a representation of an input QBF in the proposed tree-like language, which is subsequently used for model counting.
This is the first knowledge compilation based solver for \#QBF which can handle QBFs with doubly exponential model counts as well. 
Experimentally, we compare its performance with the state-of-the-art \#QBF-solver d4-QBF which works well for QBFs with small model counts. 
\end{abstract}

\section{Introduction}
In logical theory, knowledge compilation is a useful technique for dealing with various computationally hard problems, such as validity, model counting, entailment, and equivalence~\cite{DBLP:journals/jair/DarwicheM02}.
In this technique, a given logical theory (such as a CNF) is transformed into a target language (such as a decomposable negation normal form, DNNF~\cite{DBLP:journals/jacm/Darwiche01}) that can support efficient problem solving. 
Crucially, different languages allow to solve different problems efficiently, hence, the choice of a target language depends on the problems required to solve.

In this article, we study knowledge compilation for quantified Boolean formulas (QBFs).
We focus on solving important intractable queries (problems) on QBFs, for example, QBF model counting (\#QBF)~\cite{Ladner89} and QBF MaxSAT (MaxQBF)~\cite{DBLP:conf/sat/IgnatievJM13}. Recently, there has been a great interest in the \#QBF problem, both from the solving and the proof complexity perspectives. The \#QBF solver d4-QBF~\cite{CapelliLPS24} identifies different disconnected components in the QBF matrix, solves them separately and merges them at the end. It also uses several SAT solving techniques to compute the model count. The \qmice (Quantified Model-counting Induction by Claim Extension)\cite{ChedeCKS26} proof system, an extension of the \#SAT proof system \mice~\cite{fichte2022proofs},   has been recently introduced. \qmice is a line-based proof system that computes and certifies the model count of the input QBF using sound inference rules. 

Here, we propose a new target language for QBFs: Certified Truncated Quantified Tree (\cxyz, Section~\ref{sec:cert_tqt}). We show that it supports intractable queries such as \#QBF and MaxQBF efficiently with respect to the size of the representation. For \#QBF, we show that \cxyz supports computation of the model count in number of arithmetic operations bounded polynomially in terms of the representation size. Note that the model count can be doubly exponential in the number of variables of the QBF; so any algorithm computing the model count explicitly must run in exponential time in the worst case. For MaxQBF, we show that the \cxyz representation of a related QBF enables us to answer the MaxQBF query for the original QBF in time polynomial in the representation.

We also propose a DAG-like target language, called `Certified Truncated Quantified DAG' (\cdxyz, Section~\ref{sec:cert_tqd}), which generalizes \cxyz.
In addition to theoretical analysis of the proposed languages, we also provide a tool for compiling a \xyz structure from a QBF (Section~\ref{sec:tqt_compiler}).
We briefly list out our contributions below.

\paragraph*
{Our Contributions:}
\begin{enumerate}
    \item {\bf Introducing new target languages \xyz and \dxyz}: Informally, the language \xyz for a set of variables $X =\{x_1,x_2,\dots,x_n\}$ is a tuple $(Q,T)$, where $Q= Q_1x_1Q_2x_2...Q_nx_n$, $Q_i \in \{\exists,\forall\}, Q_i \neq Q_{i+1}$, and $T$ is a rooted tree where every internal node is labelled by a variable in $X$. $T$ has certain constraints : the root of $T$ is labelled by $x_1$ and all nodes at level $i$ are labelled by $x_i$. Each leaf node contains a value $c$ calculated following a fixed method (see Definition~\ref{xyz_def}). Every node has two outgoing edges, labelled with $0$ (dashed edge) and $1$ (solid edge) which are allowed to go to the same node or different nodes (i.e., parallel edges are allowed; see Figure~\ref{running_tqt_fig}).  

    For a QBF $\Phi = Q.\phi$, where $Q$ is a quantifier prefix and $\phi$ is a CNF, a \cxyz $G = (Q,T)$ is a modified version of \xyz, which certifies the equivalence of $\Phi$ and $G$. Here $T$ has additional constraints~(Definition~\ref{def:cert_xyz}): A leaf node with parent  labelled by $x_i, i\neq n$ and with a value $0$ (resp. $\geq 1$) must provide a refutation (resp. tautology proof) of the restricted QBF $\Phi|_\alpha$, where $\alpha$ is the partial assignment according to the corresponding root-to-leaf path in $T$. 
    If the two outgoing edges of a node labelled by $x_i$ are going to a single node labelled by $x_{i+1}$, then a proof of equivalence of two restricted QBFs $\Phi_0 = \Phi|_{\alpha \cup \{x_i=0\}}$ and $\Phi_1 = \Phi|_{\alpha \cup \{x_i=1\}}$ must be provided. 

    To strengthen the language \xyz we also allow edges from different nodes labelled by the same variable to go to the same child node (see Figure
    ~\ref{running_dag_fig}). We denote the resulting language as \dxyz (Definition~\ref{def:tqdag}). We denote the certified version of \dxyz as \cdxyz, which must provide the explicit proof of equivalence of the corresponding restricted QBFs (see Definition~\ref{def:ctqdag}).

    \item {\bf Analysing the strengths and limitations of the proposed languages}: We provide a detailed analysis of the strengths and weaknesses of the proposed languages. We present a family of QBFs (Definition~\ref{ctqt_ub_formula}) with doubly exponential model count but a linear size \cxyz (Theorem~\ref{thm:upper_ctqt}). We also establish limitations of the \xyz language by providing a simple quantified statement Quant-Parity$_n$ (Definition~\ref{tqt_lb_formula}) and showing that any \xyz structure for it requires size exponential in $n$ (Theorem~\ref{thm:tqt_lowr_bound}). We also show that Quant-Parity$_n$ has a polynomial sized \dxyz structure (Corollary~\ref{seperation_tqt_tqd}). As a result, \dxyz is an exponentially more succinct representation
than \xyz.

    We establish upper bound for \cdxyz. We modify the false QBF $\phi_n$ from~\cite{DBLP:journals/tcs/JanotaM15} by removing its large width clause and make it a true QBF, with doubly exponential model count, TQBF-$\phi_n$ (Definition~\ref{phi_dag_ub_formula}) and show that it has a polynomial sized \cdxyz (Theorem~\ref{thm:tqbf_phi_n_upper_bound_ctqdag}).

    \item {\bf Introducing first static proof systems for \#QBF and MaxQBF}: Given a QBF $Q.\phi$ and a corresponding \cxyz (resp. \cdxyz) representation, there is an algorithm running in time polynomial in the size of \cxyz (resp. \cdxyz) to verify the equivalence between them. 
    We further show that given a \cxyz (resp. \cdxyz) for a QBF, its model count can be computed in number of operations polynomial in the representation size. This is a pseudo-polynomial time algorithm. 
    As a result, \cxyz (resp. \cdxyz) can be seen as a sound and complete \#QBF proof system for QBFs which have exponential many model counts w.r.t. number of variables (as opposed to doubly exponential). To the best of our knowledge, these are the first static \#QBF proof systems. For the propositional model counting (\#SAT) problem, there exist static proof systems such as \kcps~\cite{Capelli19} and \cpog~\cite{bryant2023certified}. 

    Given a false QBF $Q.\phi$, the problem of MaxQBF~\cite{DBLP:conf/sat/IgnatievJM13} is to find a maximum cardinality subset of clauses $\phi' \subseteq \phi$ such that $Q.\phi'$ is true. Using the concepts of blocking variables~\cite[Definition 2.2]{DBLP:conf/sat/BonacinaBBL25}, and the \cxyz representation, we define proof systems for the MaxQBF problem (Section~\ref{maxqbf_sec}) as well.
    
    \item {\bf Comparing the strengths of \qmice and the \cdxyz proof systems}: We show that the recently introduced line-based \#QBF proof system \qmice~\cite{ChedeCKS26}, with additional semantic rules, is capable of efficiently simulating the \cdxyz proof system (Theorem~\ref{thm:qmice_p_simulate_cert_tqdag}). We further show that \qmice is exponentially stronger than \cdxyz by proving an exponential lower bound for the QBF families \textsf{QXOR-PAIRS} in \cdxyz proof system (Theorem~\ref{thm:tqd_lowr_bound}) which is known to be easy in \qmice~\cite[Section 4.1]{ChedeCKS26}.

    \item {\bf A practical \xyz compiler for QBFs}: We have implemented a solver, denoted \qbfsolver (Section~\ref{sec:tqt_compiler}) that generates a \xyz representation for an input QBF $Q.\phi$. The solver outputs the \xyz tree in a simple flat format. Then it computes the model count by running a tree evaluation algorithm on the \xyz tree. 

    We run our solver on important (existing and new) benchmarks and compare the results with the state-of-the-art \#QBF solver d4-QBF~\cite{CapelliLPS24} (Section~\ref{subsec:exp_eva}). 
    d4-QBF is known to perform well on existing benchmarks with small model counts. %However, many important QBFs in the literature have  doubly exponential counts. 
    \qbfsolver is another way of handling \#QBF specially when the count is doubly exponential. 
    We show the importance of \qbfsolver by providing two new benchmarks with doubly exponential model counts (Section~\ref{subsec:exp_eva}). 
\end{enumerate}

\section{Preliminaries}\label{prelims_sec}

Quantified Boolean Formulas (QBFs) extend propositional logic by assigning quantifiers $\exists, \forall$ to variables. A QBF is of the form $Q.F$ where $Q$ is the quantifier prefix and $F$ is a CNF matrix. $Q$ is of the form $Q_1 X_1, Q_2 X_2,...,Q_n X_n$ where $Q_i=\{\exists,\forall\}$. Variable $x\in X_i$ is said to be left of $y\in X_j$ i.e. $x \leq_Q y$ when $i\leq j$. The set of universal variables to the left of an existential variable $x$ are denoted by $\mathcal{L}_Q(x)$. QBF can be evaluated as a two-player (existential v/s universal) game, i.e. QBF is true (or false) iff existential (universal) player has a winning strategy. The game is as follows: each player assigns $0/1$ to their own variables in the sequence as in quantifier prefix, and finally after substituting all the assignments if the CNF matrix evaluates to True (or False), the existential (Universal) player wins. In this paper we use TQBF, FQBF to mean true QBF and false QBF respectively.

A generalization of QBF is a quantified statement $Q.\phi$ where the matrix $\phi$ is a propositional statement not necessarily in CNF form.
A quantified statement $Q.\phi$ can be represented semantically as a complete binary \textbf{Quantified tree} (QT) where nodes in the same level are labelled with one variable in the order of $Q$ (i.e., the root is the outermost variable, all penultimate nodes are labelled by the innermost variable and the ordering within a block can be relaxed). Two outgoing edges from every node are labelled with $0$ or $1$.
The leaves are labelled with $0/1$, which is the value that $\phi$ evaluates to when substituted with the complete assignment according to the edges on the root-to-leaf path.
The satisfiability of $Q.\phi$ is then computed in a bottom-up fashion with the semantic meaning of the variables. 
A \textbf{model-tree} (winning strategy) of $Q.\phi$ is a subtree of the Quantified-tree containing the root and one outgoing edge for nodes with an $\exists$-variable and two outgoing edges for nodes with an $\forall$-variable, and all leaves of this tree being $1$s. 
Given the Quantified-tree of a true quantified statement, we can count the number of model-trees combinatorially with a bottom-up approach: values at the leaves are already $0/1$, at the $\exists$-variables add the values of the children (because you only need one outgoing edge) and at the $\forall$-variables multiply the values of the children (because you need both outgoing edges). Finally at the root, the value equals the number of model-trees (i.e. \#QT).

We use the following running example throughout the paper. 
\begin{example}\label{running_eg}
Let $\Phi:=\forall u_1 \exists e_1 \forall u_2 \exists e_2~.~ (u_1\lor e_1\lor u_2\lor e_2)\land (u_1\lor e_1\lor u_2\lor \overline{e_2})\land (u_1\lor e_1\lor \overline{u_2}\lor e_2)\land (u_1\lor \overline{e_1}\lor u_2\lor \overline{e_2})\land (u_1\lor \overline{e_1}\lor \overline{u_2}\lor \overline{e_2})\land (\overline{u_1}\lor e_1\lor u_2\lor \overline{e_2})\land (\overline{u_1}\lor e_1\lor \overline{u_2}\lor \overline{e_2})$. The \#QBF$(\Phi)=5$, the Quantified tree of $\Phi$ is shown in Figure~\ref{running_tree} for convenience.
\end{example}
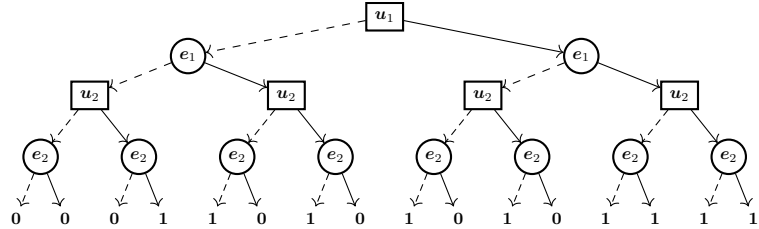
\begin{figure}[!ht]
    \centering
\begin{tikzpicture}[scale=0.65, transform shape]
\clip(-7.85,4.635) rectangle (7.85,0.125);

\node[univ, draw] (a) at (0,4.35) {$\bm u_1$};

\node[exist,draw] (b) at (-4,3.55) {$\bm e_1$};
\node[exist,draw] (c) at (4,3.55) {$\bm e_1$};

\node[univ, draw] (d) at (-6,2.75) {$\bm u_2$};
\node[univ, draw] (e) at (-2,2.75) {$\bm u_2$};
\node[univ, draw] (f) at (2,2.75) {$\bm u_2$};
\node[univ, draw] (g) at (6,2.75) {$\bm u_2$};

\node[exist,draw] (h) at (-7,1.5) {$\bm e_2$};
\node[exist,draw] (i) at (-5,1.5) {$\bm e_2$};
\node[exist,draw] (j) at (-3,1.5) {$\bm e_2$};
\node[exist,draw] (k) at (-1,1.5) {$\bm e_2$};
\node[exist,draw] (l) at (1,1.5) {$\bm e_2$};
\node[exist,draw] (m) at (3,1.5) {$\bm e_2$};
\node[exist,draw] (n) at (5,1.5) {$\bm e_2$};
\node[exist,draw] (o) at (7,1.5) {$\bm e_2$};

\node[] (p) at (-7.5,0.25) {$\bm 0$};
\node[] (q) at (-6.5,0.25) {$\bm 0$};
\node[] (r) at (-5.5,0.25) {$\bm 0$};
\node[] (s) at (-4.5,0.25) {$\bm 1$};
\node[] (t) at (-3.5,0.25) {$\bm 1$};
\node[] (u) at (-2.5,0.25) {$\bm 0$};
\node[] (v) at (-1.5,0.25) {$\bm 1$};
\node[] (w) at (-0.5,0.25) {$\bm 0$};
\node[] (x) at (0.5,0.25) {$\bm 1$};
\node[] (y) at (1.5,0.25) {$\bm 0$};
\node[] (z) at (2.5,0.25) {$\bm 1$};
\node[] (aa) at (3.5,0.25) {$\bm 0$};
\node[] (bb) at (4.5,0.25) {$\bm 1$};
\node[] (cc) at (5.5,0.25) {$\bm 1$};
\node[] (dd) at (6.5,0.25) {$\bm 1$};
\node[] (ee) at (7.5,0.25) {$\bm 1$};

\draw[black, ->] (a) -- (c); %node[pos=.35,above, right] {$0$};

\draw[black, dashed, ->] (a) -- (b); %node[pos=.35,above, left] {$1$};

\draw[black, dashed , ->] (b) -- (d); 
\draw[black, ->] (b) -- (e);
\draw[black, dashed, ->] (c) -- (f); 
\draw[black, ->] (c) -- (g);

\draw[black, dashed, dash pattern=on 2.75pt off 1.9pt, ->] (d) -- (h); 
\draw[black, ->] (d) -- (i);
\draw[black, dashed, dash pattern=on 2.75pt off 1.9pt, ->] (e) -- (j); 
\draw[black, ->] (e) -- (k);
\draw[black, dashed, dash pattern=on 2.75pt off 1.9pt, ->] (f) -- (l); 
\draw[black, ->] (f) -- (m);
\draw[black, dashed, dash pattern=on 2.75pt off 1.9pt, ->] (g) -- (n); 
\draw[black, ->] (g) -- (o);

\draw[black, dashed, dash pattern=on 2.6pt off 1.8pt, ->] (h) -- (p); 
\draw[black, ->] (h) -- (q);
\draw[black, dashed, dash pattern=on 2.6pt off 1.8pt, ->] (i) -- (r); 
\draw[black, ->] (i) -- (s);
\draw[black, dashed, dash pattern=on 2.6pt off 1.8pt, ->] (j) -- (t); 
\draw[black, ->] (j) -- (u);
\draw[black, dashed, dash pattern=on 2.6pt off 1.8pt, ->] (k) -- (v); 
\draw[black, ->] (k) -- (w);
\draw[black, dashed, dash pattern=on 2.6pt off 1.8pt, ->] (l) -- (x); 
\draw[black, ->] (l) -- (y);
\draw[black, dashed, dash pattern=on 2.6pt off 1.8pt, ->] (m) -- (z); 
\draw[black, ->] (m) -- (aa);
\draw[black, dashed, dash pattern=on 2.6pt off 1.8pt, ->] (n) -- (bb); 
\draw[black, ->] (n) -- (cc);
\draw[black, dashed, dash pattern=on 2.6pt off 1.8pt, ->] (o) -- (dd); 
\draw[black, ->] (o) -- (ee);

\end{tikzpicture}
    \caption{Quantified tree of the QBF $\Phi$ from Example~\ref{running_eg} (solid edge is $1$ and dashed edge is $0$)}
    \label{running_tree}
\end{figure}

For a propositional statement $\phi$ and a partial assignment $\alpha$ to $\vars(\phi)$,
$\phi|_\alpha$ denotes substituting values of $\vars(\alpha)$ from $\alpha$ into $\phi$. 
If $\phi|_\alpha$ is true, $\alpha$ satisfies $\phi$.  
Similarly, restricting a quantified statement $Q.\phi$ with some partial assignment $\alpha$ to $\vars_\exists(Q.\phi)$ is well defined as $Q.(\phi\land \underset{x\in\vars(\alpha)}{\bigwedge}(x=\alpha(x)))$. For restricting universal variables, we use the following.

\begin{definition}[Restriction of $\forall$-variables in Quantified statements~\cite{ChedeCKS26}]\label{restriction_qbf}

Let $Q.\phi$ be a quantified statement. 
By a restriction of the $\forall$-variable ($y =b\in\{0,1\}$) in $Q.F$ (denoted as 
$Q.\phi|_{\{y=b\}}=Q'.\phi\land (y=b)$), we mean that the quantification of $y$ is flipped to $\exists$ in the quantification sequence $Q'$ and an additional unit clause forcing the restriction $y=b$ is added to the matrix. 
\end{definition}

A proof system $\mathcal{P}$ for a language $L$ of theorems is a polynomial-time verifiable function mapping proofs from the set of $\mathcal{P}$-proofs to theorems in $L$. $\mathcal{P}$ should also be sound and complete for $L$ i.e., if a $\mathcal{P}$-proof $\pi$ exists for some theorem $x$ then $x\in L$ and if $x\in L$ then some $\mathcal{P}$-proof $\pi$ should exist for $x$. Let $\mathcal{P},\mathcal{Q}$ be proof systems for $L$. $\mathcal{P}$ p-simulates $\mathcal{Q}$ if there exists a polynomial time algorithm to convert $\mathcal{Q}$-proofs $\pi$ for $x\in L$ into $\mathcal{P}$-proofs $\pi'$ of $x$. Proof systems can be broadly classified as dynamic and static systems~\cite{GrigorievV01}. Dynamic proof systems consist of various sound inferences which can be used to gradually derive a proof, whereas in static proof systems a proof can be seen as one large inference.

\section{\fxyz (\xyz): A knowledge Compilation Language}
\label{sec:xyz}

In this section, we define a new knowledge compilation language `\fxyz (\xyz)' for quantified statements. We start with an intuitive explanation of how a quantified tree can be truncated to obtain a \fxyz.

Consider a node at depth $i$ in a quantified tree where all leaves below this node are $1$-leaves. Then, the model-tree count of the said node is easy to calculate using the algorithm discussed in Section~\ref{prelims_sec}. Hence the subtree rooted at this node can be truncated once the count is stored. Note that this count is always a power of $2$, hence, we only store the exponent to make the representation succinct. This allows us to store doubly exponential counts using linearly many bits.    

Consider another node at depth $i$ where the corresponding model-count is $0$, we truncate such nodes as they don't contribute to the satisfiability or model-count of the quantified tree.

\begin{definition}[\fxyz (\xyz)]\label{xyz_def}
A \xyz structure $G$ on variables $X = \{x_i~|~i\in [n]\}$ is a tuple $(Q,T)$ consisting of a quantifier prefix $Q$ which assigns $Q_i\in\{\exists,\forall\}$ to variables of $X$ in a sequence say $Q_1 x_1,\ldots, Q_n x_n$, and a binary tree $T$ where:
\begin{itemize}
    \item the root node is labelled with $x_1$ and nodes in subsequent levels are labelled sequentially as in $Q$ (ordering within a block can be relaxed) unless they are a leaf-node. 
    \item For a leaf-node whose parent is labelled by $x_n$, it can be labelled with either $0$ or $1$.
    \item  For a leaf-node at a depth of $i$ and $i\not= n$, the leaf label ($c$) can either be `$0$F' or can be calculated
    as follows: initially $c=0$, for $j=n$ to $i+1$, if $Q_j=\exists$ then $c=c+1$ and when $Q_j=\forall$ then $c=2\cdot c$. 
    Note that root is at depth $0$ and $c$ is the exponent discussed above when the subtree rooted at a node has all $1$-leafs in the underlying quantified tree. 
    \item each internal node has either two outgoing edges labelled with $0$ and $1$ or just a single outgoing edge represented by two multi-edges (labelled $0,1$) going to the same child node.
\end{itemize}
\end{definition}
The size of a \xyz is the number of nodes in the tree \(T\) plus the number of bits required for the integer-labels of the leaf nodes.

Observe that since all quantified statements $Q.\phi$ have a corresponding full Quantified tree (say $T$), $G=(Q,T)$ is already a valid \xyz structure corresponding to $Q.\phi$. That is, \xyz is a complete representation for quantified statements. 
For an example of \xyz, Figure~\ref{running_tqt_fig} shows a \xyz representation for the QBF $\Phi$ from (Note that for a FQBF node all leaves of the subtree in the corresponding QT need not be $0$). Also notice that even though \xyz is a tree, its definition allows multi-edges between a parent and child node. The intuition behind this is to prevent variables which do not appear in the matrix from blowing up the size of \xyz. For example, the QBF  $\exists x_1 \forall y_1...\exists x_n . (x_n)$ would need exponential sized \xyz without multi-edges.

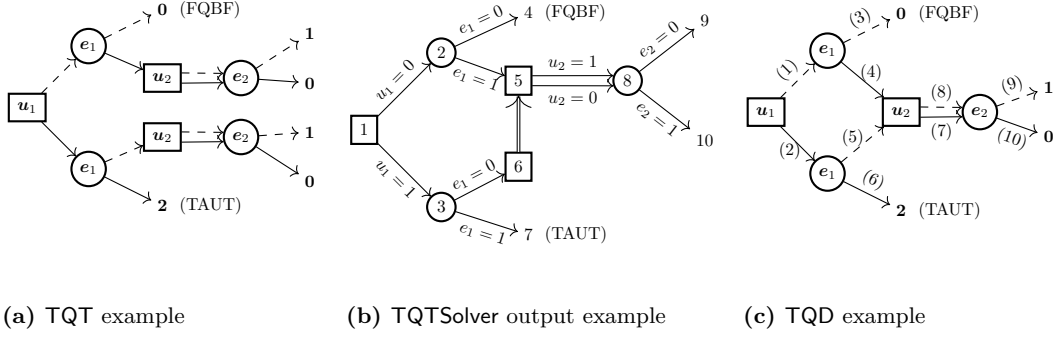
\begin{figure}[!ht]
    \centering
    \begin{subfigure}[b]{0.3\textwidth}
    \centering
\begin{tikzpicture}[scale=0.65, transform shape]
%\clip(-7.85,4.635) rectangle (7.85,0.125);

\node[univ, draw] (a) at (1.25,0) {$\bm u_1$};

\node[exist,draw] (b) at (2.5,1.25) {$\bm e_1$};
\node[exist,draw] (c) at (2.5,-1.25) {$\bm e_1$};

\node[label=right:(FQBF)] (d) at (4,2) {$\bm 0$};
\node[univ, draw] (e) at (4,0.6) {$\bm u_2$};
\node[univ, draw] (f) at (4,-0.6) {$\bm u_2$};
\node[label=right:(TAUT)] (g) at (4,-2) {$\bm 2$};

\node[exist,draw] (h) at (5.6,0.6) {$\bm e_2$};
\node[exist,draw] (i) at (5.6,-0.6) {$\bm e_2$};

\node[] (j) at (7,1.5) {$\bm 1$};
\node[] (k) at (7,0.5) {$\bm 0$};
\node[] (l) at (7,-0.5) {$\bm 1$};
\node[] (m) at (7,-1.5) {$\bm 0$};

\draw[black, ->] (a) -- (c); %node[pos=.35,above, right] {$0$};

\draw[black, dashed, ->] (a) -- (b); %node[pos=.35,above, left] {$1$};

\draw[black, dashed , ->] (b) -- (d); 
\draw[black, ->] (b) -- (e);
\draw[black, dashed, ->] (c) -- (f); 
\draw[black, ->] (c) -- (g);

\draw[black, ->] (e.-15) -- (h.195);
\draw[black, dashed, ->] (e.15) -- (h.165); 

\draw[black, ->] (f.-15) -- (i.195);
\draw[black, dashed, ->] (f.15) -- (i.165); 

\draw[black, dashed, ->] (h) -- (j); 
\draw[black, ->] (h) -- (k);
\draw[black, dashed, ->] (i) -- (l); 
\draw[black, ->] (i) -- (m);

\end{tikzpicture}
    \caption{\xyz example}
    \label{running_tqt_fig}
    \end{subfigure}
    \hfill
    \begin{subfigure}[b]{0.35\textwidth}
    \centering
\begin{tikzpicture}[scale=0.65, transform shape, node distance=2.22cm]
%\clip(-7.85,4.635) rectangle (7.85,0.125);

\node[univ] (a) {$1$};

\node[exist, above right of=a] (b) {$2$};
\node[exist, below right of=a] (c) {$3$};

\node[label=right:(FQBF), above right of=b, yshift=-0.75cm, xshift=0.2cm] (d) {$4$};
\node[univ, below right of=b, yshift=1cm] (e) {$5$};
\node[univ, above right of=c, yshift=-0.75cm] (f) {$6$};
\node[label=right:(TAUT), below right of=c, yshift=1cm, xshift=0.2cm] (g) {$7$};

\node[exist, right of=e] (h) {$8$};

\node[above right of=h, yshift=-0.35cm] (j) {$9$};
\node[below right of=h, yshift=0.35cm] (k) {$10$};

\path[->]
(a) edge node[sloped, above] {$u_1 = 0$} (b)
(a) edge node[sloped, below] {$u_1 = 1$} (c)
(b) edge node[sloped, above] {$e_1 = 0$} (d) 
(b) edge node[sloped, below] {$e_1 = 1$} (e)
(c) edge node[sloped, above] {$e_1 = 0$} (f)
(c) edge node[sloped, below] {$e_1 = 1$} (g)
(e.-20) edge node[sloped, below] {$u_2 = 0$} (h.200)
(e.20) edge node[sloped, above] {$u_2 = 1$} (h.160)
(h) edge node[sloped, above] {$e_2 = 0$} (j)
(h) edge node[sloped, below] {$e_2 = 1$} (k)
;

\draw[black, double, ->] (f) -- (e);
\end{tikzpicture}
\caption{\qbfsolver output example}
\label{fig:solver_tree_eg}
    \end{subfigure}
    \hfill
    \begin{subfigure}[b]{0.3\textwidth}
    \centering
\begin{tikzpicture}[scale=0.65, transform shape]
\clip(0.85,2.2) rectangle (7.35,-2.2);

\node[univ, draw] (a) at (1.25,0) {$\bm u_1$};

\node[exist,draw] (b) at (2.5,1.25) {$\bm e_1$};
\node[exist,draw] (c) at (2.5,-1.25) {$\bm e_1$};

\node[label=right:(FQBF)] (d) at (4,2) {$\bm 0$};
\node[univ, draw] (e) at (4,0) {$\bm u_2$};
%\node[univ, draw] (f) at (4,-0.6) {$\bm u_2$};
\node[label=right:(TAUT)] (g) at (4,-2) {$\bm 2$};

\node[exist,draw] (h) at (5.6,0) {$\bm e_2$};

\node[] (j) at (7,0.5) {$\bm 1$};
\node[] (k) at (7,-0.5) {$\bm 0$};

\draw[black, ->] (a) --node[pos=0.275, below] {${(2)}$} (c); %node[pos=.35,above, right] {$0$};

\draw[black, dashed, ->] (a) --node[above, sloped] {${(1)}$} (b); %node[pos=.35,above, left] {$1$};

\draw[black, dashed , ->] (b) --node[above, sloped] {${(3)}$} (d); 
\draw[black, ->] (b) --node[pos=0.7, above] {${(4)}$} (e);
\draw[black, dashed, ->] (c) --node[above, pos=0.275] {${(5)}$} (e); 
\draw[black, ->] (c) --node[above, sloped] {${(6)}$} (g);

\draw[black, ->] (e.-15) --node[below, sloped] {${(7)}$} (h.195);
\draw[black, dashed, ->] (e.15) -- node[above, sloped] {${(8)}$} (h.165); 

\draw[black, dashed, ->] (h) --node[above, sloped] {${(9)}$} (j); 
\draw[black, ->] (h) -- node[below, sloped] {${(10)}$} (k);

\end{tikzpicture}
    \caption{\dxyz example}
    \label{running_dag_fig}
    \end{subfigure}
\caption{Examples of \xyz and \dxyz for $\Phi$ of Example~\ref{running_eg} (solid edge is $1$ and dashed edge is $0$)}
\end{figure}

\textbf{Satisfiability in \xyz:} The satisfiability of a \xyz structure can be seen as a 2-player game between the existential and the universal player: 
A \xyz structure $G= (Q,T)$ is true (or false) if the existential (or universal) player can always land on a leaf with label other than `$0F$' (or `$0$') of $T$ when players assign values to all their respective variables in the sequence of $Q$.

\textbf{$\exists-$model trees of \xyz:} Similar to Quantified tree, model trees of TQT $G=(Q,T)$ are subtrees $T'$ of $T$, where existential nodes of $T'$ have one outgoing edge and universal nodes of $T'$ have two outgoing edges and all leaves are labelled by nodes other than `$0$F' (or `$0$'). 

Given a model tree $T'$ of TQT, $T'$ has an implicit weight $w$ associated with it and can be computed by the following: maintain a variable $val$ at every node of $T'$. In the bottom-up fashion, compute the value of $val$ at every node. At the leaves whose parent node is labelled by $x_n$, $val=$ node label itself. At the leaves whose parent is not labelled by $x_n$, $val$ is equal to $2^{\text{node-label}}$.
At intermediate nodes, if the label is an $\exists$-variable, $val$ is the $val$ of this child node and if the label is an $\forall$-variable, $val$ is the product of $val$s of children nodes. Note that in case of multi-edges, we consider the single child node twice in this computation. Finally at the root we have weight of the model-tree $T'$.

\textbf{\#\xyz:} Given a \xyz $G=(Q,T)$, let $\mathcal{M}$ contain the set of all distinct model-trees $T_i$ of $G$. The \#\xyz is the sum of the weights of all model-trees $T_i$ of $T$. Observe that this \#\xyz is same as \#QT discussed in Section~\ref{prelims_sec}. 

Observe that \#\xyz answer can also be computed directly from the \xyz
structure by following the model counting algorithm of QT and simply assuming the value corresponding to the leaves labelled with $0$F as $0$, $1$-leaves at depth $n$ as $1$ and all other leaves as $2^{<\text{node-label}>}$.

\textbf{Verifying that a given structure $(Q,T)$ is a valid \xyz:} Verify that in every root-to-leaf path, the nodes are labelled in sequence according to $Q$. For every non $0F$ leaf node whose parent is not labelled by $x_n$, verify that the label actually matches the computable value from Definition~\ref{xyz_def}.

\textbf{Tractable Queries on \xyz}: Many queries like the satisfiability, \#\xyz, Conjunction of two \xyz with the same quantifier sequence, Ordered Conjunction of two \xyz with disjoint variables, and Conditioning of $\exists$-variables in \xyz are tractable. See Appendix~\ref{sec:query_support_tqt} for the details.

\section{\fcxyz for QBFs}\label{sec:cert_tqt}
In this section, we define certified versions of \xyz denoted as \cxyz. Inspired from~\cite{CapelliLPS24},  the idea is to include additional information at the nodes that can be used to verify the equivalence between the given QBF and the \cxyz. 

\begin{definition}[\fcxyz]\label{def:cert_xyz}
    Fix a QBF proof system $\mathcal{P}$. 
    A \cxyz consists of a \xyz structure $G=(Q,T)$ along with the corresponding QBF $Q.F$. The tree $T$ has the following additional labels:
    \begin{itemize}
        \item For a leaf node with label `\ $0F$' whose parent is labelled by $x_i$ and $i\not= n$, a false QBF $\mathcal{P}$-proof  of $Q.F$ restricted by the the values according to the root-to-leaf path to variables $x_1,...,x_i$ is included.
        \item  For a leaf node with a label other than `\ $0F$', whose parent is labelled by $x_i$ and $i\not= n$, a Tautology proof of $Q.F$ restricted by the the values according to the root-to-leaf path to variables $x_1,...,x_i$ is included.

        \item For a node labelled by $x_i$ whose both outgoing edges are going to the same child node labelled by $x_{i+1}$: let the assignment from the root to $x_i$ be $\alpha$. A proof of equivalence of the two resulting QBFs $\Phi_0,\Phi_1$ obtained by restricting $Q.F$ with $\alpha\cup\{x_i=0/1\}$. That is, $\Phi_0=Q.F|_{\alpha\cup\{x_i=0\}}$ and $\Phi_1=Q.F|_{\alpha\cup \{x_i=1\}}$. Equivalence proof of these QBFs (denoted $\Phi_0\Leftrightarrow_{Sk}\Phi_1$), proved by an FQBF $\mathcal{P}$-proof of a special encoding defined in~\cite{PfeifferGS25}, is included. 
    \end{itemize}
\end{definition}
The length of a \cxyz is the number of nodes in it plus the number of bits required to represent the leaf labels (or equivalently, the size of its underlying \xyz). The size of a \cxyz is its length plus the size of all the false QBFs proofs, tautology proofs, and the QBFs equivalence proofs that are part of it. 

\textbf{Verifying correctness of a \cxyz:} Given a \cxyz $G=(Q,T)$ and a QBF $Q.F$, to verify that $G\leftrightarrow Q.F$: For every $0/1$-leaf in $G$ with parent $x_n$, take the root-to-leaf assignment $\alpha$ and verify that $F|_\alpha=0/1$ accordingly. Verify all the tautology and FQBF and equivalence proofs in $G$ with the corresponding QBF restrictions. 

All queries tractable in \xyz are also tractable in \cxyz. 
Also note that \#\cxyz value answers the \#QBF question of the corresponding QBF.

\subsection{Solving MaxQBF problem using \cxyz}\label{maxqbf_sec}
Now we show how to solve the MaxQBF problem by using \cxyz. 
The problem of maximum satisfiability (MaxSAT) is to find an assignment that maximizes the number of satisfied clauses of a given input unsatisfiable CNF formula. Similarly the problem of MaxQBF is, given a false QBF $Q.F$ as input, to find a largest subset of clauses $F'\subseteq F$ such that $Q.F'$ is a true QBF.

The MaxQBF problem is tractable in \cxyz for which the input QBF is transformed into a blocking variable encoding format~\cite{DBLP:conf/stacs/CapelliM19,DBLP:conf/sat/BonacinaBBL25} that we describe now. Given a false QBF $\Phi:= Q~.~C_1\land C_2\land...\land C_m$, let $b_1,...,b_m$ be new variables $\notin\vars(\Phi)$. The blocking variable encoding of $\Phi$ is the following QBF:
$\Phi':= \exists b_1,...,b_m~Q~.~(C_1\lor b_1)\land...\land(C_m\lor b_m)$.

Note that $\Phi'$ is now a True QBF and any winning strategy of $\Phi'$ would also consist of an assignment to variables $b_1,...,b_m$. We define the cost of a winning strategy $\mathcal{S}$ in $\Phi'$, denoted as $cost(\mathcal{S})$, as $\Sigma_{i\in[m]} b_i$.
Now given $\Phi'$, the minimum Hamming weight of $\Phi'$ is defined as follows:
$\Min_{HW}(\Phi',\{b_1,...,b_m\}):=\underset{\mathcal{S}\models\Phi'}{\min}(cost(\mathcal{S}))$

Notice that MaxQBF$(\Phi)=m-\Min_{HW}(\Phi',\{b_1,...,b_m\})$. Below, we show that minimum Hamming weight is easy to compute in a \cxyz structure corresponding to the QBF $\Phi'$.

Let $G=(Q',T)$ be a \cxyz structure of $\Phi'$. The algorithm to compute the $\Min_{HW}(\Phi',\{b_1,...,b_m\})$ is as follows: Maintain a counter variable $val$ at every node of $T$. We compute the value of $val$ at every node in a bottom-up fashion. The idea is to store the minimum cost possible for a model-tree rooted at every node of $T$.
At the $0$-leaves (or $1$-leaves) whose parent is labelled by the innermost variable of $Q$, $val=\infty$ (or $0$). Similarly, at a tautological leaf node, the $val=0$ and at an FQBF node, the $val=\infty$. Now recursively, at every node labelled by a non-blocking  existential-variable, the $val$ is the minimum of $val$ from its children nodes. On the other hand, at every node labelled by a non-blocking  universal-variable, the $val$ is the maximum of $val$ from its children nodes. In the recursion, if a node is labelled by a blocking variable, its $val=\min(0+val(child_0), 1+val(child_1))$.
Finally the $val$ at the root is the minimum Hamming weight answer required.

\textbf{Static Proof systems for \#QBF:}
One can use \cxyz to define a proof system for the \#QBF problem. Given a QBF $Q.F$, a proof of the fact that $Q.F$ has $k$ winning strategies is: A \cxyz structure $G=(Q,T)$ corresponding to $Q.F$, where $G$ has $k$ winning strategies.

\textbf{Static proof systems for MaxQBF:}
One can also use \cxyz to define a proof system for the MaxQBF problem. Given a QBF $Q.F$, a proof of the fact that MaxQBF answer of $Q.F$ is $k$ is: A blocking variable encoded QBF $\Phi'=Q'.F'$, a \cxyz structure $G=(Q',T)$ corresponding to $\Phi'$, where $T$ has a minimum Hamming weight of $k$.

\subsection{Upper bounds for cert-Truncated Quantified Trees}\label{sec:applications}

Note that if we pick a strong FQBF system $\mathcal{P}$ (i.e. proving exponential lower bounds in $\mathcal{P}$ is open), any False QBF $\Phi$ which has short proofs in $\mathcal{P}$ also has a small \cxyz structure which is basically a single FQBF root node accompanied by the easy $\mathcal{P}$-proof of falsity.

For a non-trivial upper bound, consider the following true QBF:

\begin{definition}\label{ctqt_ub_formula}
    (\textsf{Alternate}$_n$) \textsf{Alt$_n$} $:= \exists e_1\forall u_1\exists e_2...\forall u_{n-1}\exists e_n~.~(\overline{e_1}\lor \overline{u_1})\land \\\hspace*{5cm}\underset{i\in[2,n-1]}{\bigwedge}(\overline{u_1}\lor ...\lor\overline{u_{i-1}}\lor\overline{e_i}\lor\overline{u_i})~\land (\overline{u_1}\lor ...\lor\overline{u_{n-1}}\lor e_n)$
\end{definition}

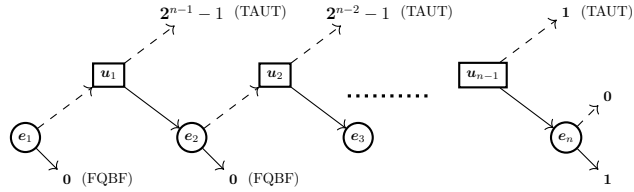
\begin{figure}[!ht]
    \centering
\begin{tikzpicture}[scale=0.55, transform shape]
\clip(-8.5,3.5) rectangle (6.6,-1.5);

\node[exist, draw] (a) at (-8,0) {$\bm e_1$};

\node[univ,draw] (b) at (-6,1.5) {$\bm u_1$};
\node[label=right:(FQBF)] (c) at (-7,-1) {$\bm 0$};

\node[label=right:(TAUT)] (d) at (-4,3) {\Large$\bm 2^{n-1}-1$};

\node[exist, draw] (e) at (-4,0) {$\bm e_2$};

\node[univ,draw] (f) at (-2,1.5) {$\bm u_2$};
\node[label=right:(FQBF)] (g) at (-3,-1) {$\bm 0$};

\node[label=right:(TAUT)] (h) at (0,3) {\Large$\bm 2^{n-2}-1$};
\node[exist, draw] (i) at (0,0) {$\bm e_3$};

\draw[black, very thick, dotted, -] (-0.25,1) -- (1.75,1); 

\node[univ,draw] (j) at (3,1.5) {$\bm u_{n-1}$};
\node[label=right:(TAUT)] (k) at (5,3) {$\bm 1$};
\node[exist, draw] (l) at (5,0) {$\bm e_n$};

\node[] (m) at (6,1) {$\bm 0$};
\node[] (n) at (6,-1) {$\bm 1$};

%\node[] (ee) at (7.5,0.25) {$\bm 1$};
%node[pos=.35,above, right] {$0$}

\draw[black, dashed, ->] (a) -- (b);
\draw[black, ->] (a) -- (c);  

\draw[black, dashed, ->] (b) -- (d); 
\draw[black, ->] (b) -- (e); 

\draw[black, dashed, ->] (e) -- (f); 
\draw[black, ->] (e) -- (g); 

\draw[black, dashed, ->] (f) -- (h); 
\draw[black, ->] (f) -- (i); 

\draw[black, dashed, ->] (j) -- (k); 
\draw[black, ->] (j) -- (l);

\draw[black, dashed, ->] (l) -- (m); 
\draw[black, ->] (l) -- (n);

\end{tikzpicture}
    \caption{\xyz for \textsf{Alt$_n$} from Definition~\ref{ctqt_ub_formula} (solid edge is $1$ and dashed edge is $0$)}
    \label{upperbound_tqt}
\end{figure}

\begin{theorem}\label{thm:upper_ctqt}
\textsf{Alt$_n$} has a linear-sized \cxyz representation.
\end{theorem}

\begin{proof}
    Figure~\ref{upperbound_tqt} shows a \xyz for \textsf{Alt$_n$}. The Tautology proofs are all easy here. The FQBF proofs are also easy to see, for an intuition, in the first FQBF node, restricting \textsf{Alt$_n$} by $e_1=1$ results in a immediate winning strategy for the universal player by just setting $u_1=1$. Similar case arises in the remaining FQBFs as well.

The size of this \xyz is $4n-1$ nodes as compared to exponential nodes in the quantified tree. All accompanying proofs are also small in sizes. Therefore \textsf{Alt$_n$} has a linear size \cxyz structure.
\end{proof}

\subsection{Lower bounds for Truncated Quantified Trees}
\label{sec:lowerbound}

Quantified statements which are true but hard to find tautologies/ FQBFs/ equivalences at higher levels in the Quantified tree would also need exponential sized \xyz structures. For example, consider the following parity based quantified statement. 

\begin{definition}[Quant-Parity$_n$]\label{tqt_lb_formula}
   $\exists e_1\forall u_1...\forall u_{n-1}\exists e_n~.~(e_1\oplus u_1\oplus ... \oplus u_{n-1}\oplus e_n)$ 
\end{definition}

\begin{theorem}\label{thm:tqt_lowr_bound}
    Any \xyz representing \textsf{Quant-Parity$_n$} requires size exponential in $n$. 
\end{theorem}
\begin{proof}
Consider any node of the full quantified tree for \textsf{Quant-Parity$_n$}, it corresponds to a partial restriction of the variables as per the path taken from the root to reach this node.
Consider any such partial restriction \(\rho\) which doesn't contain all the variables.
Under this restriction, the matrix remains a parity of the variables not contained in \(\rho\).
Hence, the formula is not a Tautology at any of the internal nodes.
Moreover, the existential player makes the last move, i.e. to choose the value of \(e_n\).
Any partial restriction obtained from the quantified tree does not restrict \(e_n\), allowing the existential player to win by choosing an appropriate value for \(e_n\).
Therefore, the restricted quantified statement is also not False at any internal node of the quantified tree.

Finally, the branches originating from any internal node set different values for the variable considered at the said node.
Hence, the corresponding quantified statements are different on the two child branches of the node.
This implies that the two child sub-trees of any internal node cannot be equivalent as QBFs.
Therefore, any TQT for \textsf{Quant-Parity$_n$} must expand fully, and be of exponential size.
\end{proof}

\section{Certified Truncated Quantified DAG (\cdxyz) language}\label{sec:cert_tqd}
In this section, we generalize the \xyz, \cxyz structures to DAG-like structures (\dxyz, \cdxyz respectively). Precisely,

\begin{definition}[\fdxyz]\label{def:tqdag}
A \dxyz is a generalization of the \xyz structure $G=(Q,T)$ where $T$ is additionally allowed to have the following:
\begin{itemize}
    \item two or more outgoing edges from the nodes labelled by the same variable can go to the same child node. We call these the DAG-like edges.
\end{itemize}
\end{definition}

\begin{definition}[\fcdxyz]\label{def:ctqdag}
Fix a QBF proof system $\mathcal{P}$. A \cdxyz consists of a \dxyz structure $G=(Q,T)$ along with the corresponding QBF $Q.F$ and the DAG $T$ has labels as defined in Definition~\ref{def:cert_xyz} and the following additional labels:
\begin{itemize}
    \item For all the DAG-like edges (say $j$) going to the same child node labelled $x_{i+1}$: let the assignments from the root to $x_{i+1}$ be $\alpha_1,...,\alpha_j$. A proof of equivalence of the $j$ resulting QBFs $\Phi_{\alpha_1},...,\Phi_{\alpha_j}$ when restricting $Q.F$ with $\alpha_1,...,\alpha_j$ respectively is included. Equivalence of these QBFs (denoted $\Phi_1\Leftrightarrow_{Sk}...\Leftrightarrow_{Sk}\Phi_j$) can be proved by $O(j-1)$ FQBF $\mathcal{P}$-proofs of a special encoding defined in~\cite{PfeifferGS25}.
\end{itemize}
Note that after equivalence proofs are attached to a DAG-node $v$, considering any one of its root-to-$v$ paths would suffice when proving equivalence for further DAG-nodes in the subtree rooted at $v$.
\end{definition}

\subsection{Upper bounds for \cdxyz}\label{up_cqbf_sec}
In this section, we show the strength of \cdxyz by providing explicit upper bounds. For example, the quantified statement 
\textsf{Quant-Parity$_n$} from Definition~\ref{tqt_lb_formula} can be encoded as a QBF using Tseitin variables (Definition~\ref{parity_dag_ub_formula}) are shown to be easy for \cdxyz (Theorem~\ref{thm:upperbound_ctqdag}).

\begin{definition}[TQBF-Parity$_n$]\label{parity_dag_ub_formula}

$\Phi:=\exists e_1 \forall u_1 \exists t_1 \exists e_2 \forall u_2 \exists t_2 ... \exists t_{n-1} \exists e_n ~.~ \\\hspace*{3cm}(t_1\leftrightarrow e_1\oplus u_1)\land  ~\underset{i\in[2,n-1]}{\bigwedge} (t_i\leftrightarrow t_{i-1}\oplus e_i\oplus u_i)~~\land (t_{n-1}\oplus e_n)$
\end{definition}

\begin{figure}[!ht]
    \centering
\begin{tikzpicture}[scale=0.6, transform shape]
\clip(-6.5,2) rectangle (12.25,-2);

\node[exist, draw] (a) at (-6,0) {$\bm e_1$};

\node[univ,draw] (b) at (-4.5,1.5) {$\bm u_1$};

\node[univ,draw] (c) at (-4.5,-1.5) {$\bm u_1$};

\node[exist,draw] (d) at (-2,1.5) {$\bm t_1$};

\node[exist,draw] (e) at (-2,-1.5) {$\bm t_1$};

\node[label=right:(FQBF)] (f) at (-0.5,0) {$\bm 0$};

\node[exist,draw] (g) at (1,1.5) {$\bm e_2$};

\node[exist,draw] (h) at (1,-1.5) {$\bm e_2$};

\node[univ,draw] (i) at (4,1.5) {$\bm u_2$};

\node[univ,draw] (j) at (4,-1.5) {$\bm u_2$};

\node[exist,draw] (k) at (7,1.5) {$\bm t_2$};

\node[exist,draw] (l) at (7,-1.5) {$\bm t_2$};

\node[label=right:(FQBF)] (m) at (8.25,0) {$\bm 0$};

\draw[dotted, thick] (8,1.5) -- (9,1.5);
\draw[dotted, thick] (8,-1.5) -- (9,-1.5);

\node[exist,draw] (n) at (10,1.5) {$\bm e_n$};

\node[exist,draw] (o) at (10,-1.5) {$\bm e_n$};

\node[] (p) at (12,1.5) {$\bm 0$};
\node[] (q) at (12,-1.5) {$\bm 1$};

% \node[label=above:(TAUT)] (d) at (-4,4) {$\bm 4k^2$};

% \node[exist, draw] (e) at (-4,0) {$\bm e_2$};

% \node[univ,draw] (f) at (-2,2) {$\bm u_2$};
% \node[label=below:(FQBF)] (g) at (-2,-2) {$\bm 0$};

% \node[label=above:(TAUT)] (h) at (0,4) {$\bm k$};
% \node[exist, draw] (i) at (0,0) {$\bm e_3$};

% \draw[black, very thick, dotted, -] (1,1) -- (3,1); 

% \node[univ,draw] (j) at (4,2) {$\bm u_{n-1}$};
% \node[label=above:(TAUT)] (k) at (6,4) {$\bm 4$};
% \node[exist, draw] (l) at (6,0) {$\bm e_n$};

%\node[] (ee) at (7.5,0.25) {$\bm 1$};
%node[pos=.35,above, right] {$0$}

\draw[black, dashed, ->] (a) -- (b);
\draw[black, ->] (a) -- (c);  

\draw[black, dashed, ->] (b) -- (d);
\draw[black, ->] (b) -- (e);  

\draw[black, dashed, ->] (c) -- (e);
\draw[black, ->] (c) -- (d);  

\draw[black, dashed, ->] (d) -- (g);
\draw[black, ->] (d) -- (f);  

\draw[black, dashed, ->] (e) -- (f);
\draw[black, ->] (e) -- (h);  

\draw[black, dashed, ->] (g) -- (i);
\draw[black, ->] (g) -- (j);  

\draw[black, dashed, ->] (h) -- (j);
\draw[black, ->] (h) -- (i);  

\draw[black, dashed, ->] (i) -- (k);
\draw[black, ->] (i) -- (l);  

\draw[black, dashed, ->] (j) -- (l);
\draw[black, ->] (j) -- (k);  

\draw[black, ->] (k) -- (m); 
\draw[black, dashed, ->] (l) -- (m);

\draw[black, dashed, ->] (n) -- (p);
\draw[black, ->] (n) -- (q);  

\draw[black, dashed, ->] (o) -- (q);
\draw[black, ->] (o) -- (p);  

\end{tikzpicture}
    \caption{\dxyz for \textsf{TQBF-Parity$_n$} (solid edge is $1$ and dashed edge is $0$)}
    \label{upperbound_dag1}
\end{figure}

Observe that it is not hard to extend the argument in Theorem~\ref{thm:tqt_lowr_bound} to show that \textsf{TQBF-Parity$_n$} also needs exponential sized \xyz representations.

\begin{theorem}\label{thm:upperbound_ctqdag}
\textsf{TQBF-Parity$_n$} has a polynomial-sized \cdxyz structure.
\end{theorem}
\begin{proof}
Figure~\ref{upperbound_dag1} shows a \dxyz for \textsf{TQBF-Parity$_n$}. The FQBF proofs are easy to see as they violate the Tseitin definition clauses.
Also the DAG-like edges result in identical restricted QBFs so no additional proofs of equivalence are necessary.

The size of this \dxyz is $6n -2$ nodes and the accompanying proofs are also trivially small in sizes. Therefore \textsf{TQBF-Parity$_n$} has a linear size \cdxyz structure.
\end{proof}

Observe that the \dxyz shown in Figure~\ref{upperbound_dag1} can be modified by removing the $t$-variables and connecting the outgoing edges of the $u_i$ nodes directly to the corresponding $e_{i+1}$ nodes. The resulting construction naturally generalizes to a compact \dxyz representation of the \textsf{Quant-Parity$_n$} formulas from Definition~\ref{tqt_lb_formula}. In Theorem~\ref{thm:tqt_lowr_bound}, we show that these formulas need exponential \xyz representations. Therefore we have the following:

\begin{corollary}\label{seperation_tqt_tqd}
\dxyz (\cdxyz) is strictly more succinct than \xyz (\cxyz).
\end{corollary}

Below we show another upper bound for \cxyz for a family of QBFs which we use as benchmarks in our solver evaluation (see Section~\ref{subsec:exp_eva}). 

\begin{definition}[TQBF-$\phi_n$]\label{phi_dag_ub_formula}

$\phi_n := \exists e_1 \forall u_1 \exists c_1, c_2 ... \exists e_n \forall u_n \exists c_{2n-1}, c_{2n}~.$
$$\underset{i\in[n]}{\bigwedge} (\overline{e_i}\lor c_{2i-1})\land (\overline{u_i}\lor c_{2i-1})\land (e_i \lor c_{2i})\land (u_i\lor c_{2i})
$$
\end{definition}

\begin{theorem}\label{thm:tqbf_phi_n_upper_bound_ctqdag}
\textsf{TQBF-$\phi_n$} has a polynomial-sized \cdxyz representation.
\end{theorem}
\begin{proof}
Figure~\ref{upperbound_dag2} shows a \dxyz for \textsf{TQBF-$\phi_n$}. The FQBF proofs are easy to see here as they violate one or more input clauses. For equivalence of two child nodes of a node (i.e. double edges between two nodes in Figure~\ref{upperbound_dag2}), the variable of the parent node doesn't appear in the corresponding restricted QBF making the corresponding restricted QBFs of the child nodes identical. Similarly, at the other DAG-like edges also the restricted QBFs are provably identical. Note that even though there are more than 1 root-to-node pairs for such nodes, it is enough to show that one of them are equivalent.

The size of this \dxyz is $10n+2$ nodes and the accompanying proofs are trivial. Therefore \textsf{TQBF-$\phi_n$} has a linear size \cdxyz structure.
\end{proof}

\begin{figure}[!ht]
    \centering
\begin{tikzpicture}[scale=0.55, transform shape]
\clip(-2.9,3.6) rectangle (17.25,-3.1);

\node[exist, draw] (a) at (-2.5,0) {$\bm e_1$};

\node[univ,draw] (b) at (-0.75,1.5) {$\bm u_1$};

\node[univ,draw] (c) at (-0.75,-1.5) {$\bm u_1$};

\node[exist,draw] (d) at (1.25,3) {$\bm c_1$};

\node[exist,draw] (e) at (1.25,0) {$\bm c_1$};

\node[exist,draw] (f) at (1.25,-2.5) {$\bm c_1$};

\node[exist,draw] (g) at (4.25,3) {$\bm c_2$};

\node[exist,draw] (h) at (4.25,0) {$\bm c_2$};

\node[exist,draw] (i) at (4.25,-2.5) {$\bm c_2$};

\node[exist,draw] (j) at (6.25,0) {$\bm e_2$};

%\node[label=275:(FQBF)] (f) at (-1,0) {$\bm 0$};

\draw[dotted, thick] (6.75,0) -- (8,0);
%\draw[dotted, thick] (9,-2) -- (10,-2);

\node[exist, draw] (k) at (8.5,0) {$\bm e_n$};

\node[univ,draw] (l) at (10,1.5) {$\bm u_n$};

\node[univ,draw] (m) at (10,-1.5) {$\bm u_n$};

\node[exist,draw] (n) at (12,3) {$\bm c_{2n-1}$};

\node[exist,draw] (o) at (12,0) {$\bm c_{2n-1}$};

\node[exist,draw] (p) at (12,-2.5) {$\bm c_{2n-1}$};

\node[exist,draw] (q) at (15,3) {$\bm c_{2n}$};

\node[exist,draw] (r) at (15,0) {$\bm c_{2n}$};

\node[exist,draw] (s) at (15,-2.5) {$\bm c_{2n}$};

\node[label=180:(FQBF)] (t) at (2.25,1.5) {$\bm 0$};
\node[label=90:(FQBF)] (u) at (2.5,-1.5) {$\bm 0$};
%\node[label=90:(FQBF)] (v) at (2,3.5) {$\bm 0$};

\node[label=180:(FQBF)] (v) at (13,1.5) {$\bm 0$};
\node[label=90:(FQBF)] (w) at (14,-1.5) {$\bm 0$};

%\node[] (x) at (17,2) {$\bm 0$};
\node[] (y) at (17,0) {\Large$\bm 1$};

\draw[black, dashed, ->] (a) -- (b);
\draw[black, ->] (a) -- (c);  

\draw[black, dashed, ->] (b) -- (d);
\draw[black, ->] (b) -- (e);  

\draw[black, dashed, ->] (c) -- (e);
\draw[black, ->] (c) -- (f);  

\draw[black, dashed, ->] (d.-18) -- (g.198);
\draw[black, ->] (d.15) -- (g.165);  

\draw[black, dashed, ->] (e) -- (t);
\draw[black, ->] (e) -- (h);  

\draw[black, dashed, ->] (f) -- (u);
\draw[black, ->] (f) -- (i);  

\draw[black, dashed, ->] (h) -- (t);
\draw[black, ->] (h) -- (j);  

\draw[black, dashed, ->] (g) -- (t);
\draw[black, ->] (g) -- (j);  

\draw[black, dashed, ->] (i.15) -- (j.275);
\draw[black, ->] (i) -- (j.225); 

\draw[black, dashed, ->] (k) -- (l);
\draw[black, ->] (k) -- (m);  

\draw[black, dashed, ->] (l) -- (n);
\draw[black, ->] (l) -- (o);  

\draw[black, dashed, ->] (m) -- (o);
\draw[black, ->] (m) -- (p);  

\draw[black, dashed, ->] (n.-15) -- (q.198);
\draw[black, ->] (n.10) -- (q.165);  

\draw[black, dashed, ->] (o) -- (v);
\draw[black, ->] (o) -- (r); 

\draw[black, dashed, ->] (p) -- (w);
\draw[black, ->] (p) -- (s); 

\draw[black, dashed, ->] (q) -- (v);
\draw[black, ->] (q) -- (y); 

\draw[black, dashed, ->] (r) -- (v);
\draw[black, ->] (r) -- (y);

\draw[black, dashed, ->, shorten >=1pt] (s.-15) -- (y.275);
\draw[black, ->] (s.40) -- (y.225);  

\end{tikzpicture}
    \caption{\dxyz for \textsf{TQBF-$\phi_n$} (solid edge is $1$ and dashed edge is $0$)}
    \label{upperbound_dag2}
\end{figure}

\subsection{Lower bounds for \dxyz}
The QBF formulas \textsf{QXOR-PAIRS} (Definition~\ref{qxor_dag_lb_formula}) is shown to be easy for \qmice system in ~\cite{ChedeCKS26}. 
We show that these formulas need exponential sized \dxyz representations (Theorem~\ref{thm:tqd_lowr_bound}).

\begin{definition}[\textsf{QXOR-PAIRS}~\cite{ChedeCKS26}]\label{qxor_dag_lb_formula}
$\mathcal{Q}.\oplus_n := \forall u_1, ..., u_n \underset{i\neq j\in[n]}{\exists} x_{i,j}~.~ \underset{i\neq j\in[n]}{\bigwedge} (x_{i,j}\leftrightarrow u_i\oplus u_j)$
\end{definition}

\begin{theorem}\label{thm:tqd_lowr_bound}
    Any \dxyz representing \textsf{QXOR-PAIRS} requires size exponential in $n$. 
\end{theorem}
\begin{proof}
The formulas of \textsf{QXOR-PAIRS} have a unique winning strategy, hence no internal node of the full quantified tree of $\mathcal{Q}.\oplus_n$ can be a Tautology. 
There are no possible FQBFs or equivalent nodes in the top half of the quantified tree corresponding to universal variables $u_1,...,u_n$ as: 1) every possible assignment to these variables has a unique assignment to the existential variables to make the formula True. So none of these internal nodes can ever be a FQBF. 2) consider two root-to-node paths where one $u_k$ variable is assigned differently. In one of the restricted QBFs, there are clauses $(x_{i,k}\leftrightarrow \overline{u_i})$ and the other has clauses $(x_{i,k}\leftrightarrow u_i)$ for $i\in [n]\setminus k$. Both have their own winning strategies and both are TQBFs, however these strategies as one can see are not equivalent. The argument can be extended to the paths differing by more than one assignment, in which case, there will simply be more such clauses to witness the difference.

Hence, any \dxyz of $\mathcal{Q}.\oplus_n$ has to include all $2^n$ nodes in the top half of the DAG using labels $u_1,..,u_n$. 
\end{proof}

Observe that, in fact, any QBF whose matrix requires remembering the individual truth values of a linear number of variables provides a lower bound for the \dxyz structure. This includes the \textsf{QXOR-PAIRS} formulas from~\cite{ChedeCKS26} as discussed in Theorem~\ref{thm:tqd_lowr_bound} and also the Indexed Affine formulas from the same paper.

\section{\qmice versus \cdxyz}\label{qmice_simulation_sec}
In this section, we show that \qmice~\cite{ChedeCKS26} with an equivalence rule can p-simulate \cdxyz proof system. Since the \textsf{QXOR-PAIRS} from~\cite{ChedeCKS26} requires exponential size proofs in \cdxyz (Theorem~\ref{thm:tqd_lowr_bound}) and it is known to be easy for \qmice~\cite{ChedeCKS26}, we conclude that \qmice is exponentially stronger than \cdxyz. Before describing the \qmice p-simulation result, we first discuss the rules of \qmice relevant to this paper.

Proofs in \qmice are built with tuples of the format $(Q.F|_\rho,~ \alpha,~ c)$ called claims where $\rho,\alpha$ are partial assignments to universal and existential variables respectively. A correct claim certifies that the QBF $Q.F|_{\rho\cup\alpha}$ has $c$ number of winning strategies for existential variables not assigned in $\alpha$. In general $\alpha$ can be strategies to existential variables, however in this paper, we only need $\alpha$ to be a partial assignment to existential variables. The proof begins with an Axiom claim $(Q.F|_\rho,~ \alpha,~ c)$ where $F|_{\rho\cup\alpha}$ is a Tautology and the value of $c$ in this case is computed the same way as in this paper however they store the exact count instead of just the exponent. 

Restricted versions of \qmice inference rules that we use in this paper are: 
\begin{itemize}
    \item Join rule: {\Large$\frac{(Q.F|_{\rho\cup y=0},~ \alpha,~ c_1)~~~(Q.F|_{\rho\cup y=1},~ \alpha,~ c_2)}{(Q.F|_\rho,~ \alpha,~ c_1\cdot c_2)}$} can be used when $\rho\cup\alpha$ assign all and only variables to the left of $y$ in the prefix $Q$.
    \item Composition-c rule: {\Large$\frac{(Q.F|_{\rho},~ \alpha \cup y=0,~ c_1)~~~(Q.F|_{\rho},~ \alpha \cup y=1,~ c_2)}{(Q.F|_\rho,~ \alpha,~ c_1+ c_2)}$} can be used when $\rho\cup\alpha$ assign all and only variables to the left of $y$ in the prefix $Q$.
    \item Absence of models rule: {\Large$\frac{}{(Q.F|_\rho,~ \alpha,~ 0)}$} can be used when there is explicit proof that $Q.F|_{\rho\cup\alpha}$ is a False QBF.
\end{itemize}

Since \qmice was introduced as a modular system where other sound rules can be added when required. We introduce the following equivalence rule: {\Large$\frac{(Q.F|_{\rho_1},~ \alpha_1,~ c)}{(Q.F|_{\rho_2},~ {\alpha_2},~ c)}$} whenever there is explicit proof that the QBFs are Skolem equivalent i.e. $Q.F|_{\rho_1\cup\alpha_1} \Leftrightarrow_{Sk} Q.F|_{\rho_2\cup\alpha_2}$ as defined in~\cite{PfeifferGS25}. Addition of similar rules (semantic reform rule~\cite[Figure 4]{BeyersdorffHK26}) for the propositional MICE system has been studied and shown to be exponentially stronger than the standard MICE system.

\begin{theorem}\label{thm:qmice_p_simulate_cert_tqdag}
    \qmice along with an equivalence rule p-simulates \cdxyz proof system.
\end{theorem}
\begin{proof}
Consider a \cdxyz structure $G=(Q,T)$ with all corresponding proofs of equivalence, FQBFs and Tautologies representing a QBF $Q.F$. Below we build the \qmice proof of $Q.F$ in a bottom-up fashion from $G$:
\begin{itemize}
    \item For a leaf that is either a TAUT node or a $1$-node, consider one of its root-to-leaf path and use the corresponding assignment to derive the Axiom claim with the count as $2^{<node-label>}$ for a TAUT node and $1$ for a $1$-node. 
    \item Similarly, for a leaf node that is either a FQBF or a $0$-node, use the corresponding assignment of one of its root-to-leaf path to derive the absence of models claim.
    \item For an internal node labelled with an $\exists$-variable, use the composition rule to derive it's corresponding claim once the child claims are derived.
    \item Similarly, for an internal node labelled with an $\forall$-variable, use the composition rule to derive it's corresponding claim once the child claims are derived.
    \item After deriving the corresponding claim of a node, if it has more than one incoming edges, pick the parent whose root-to-node path doesn't match the generated claim. Use the equivalence proofs to derive a copy of the claim for one of the root-to-node paths involving the said parent.
\end{itemize}
Finally at the root, you should have methodically dropped all restrictions and derived the final claim $(Q.F,~ \emptyset,~ c)$ where $c$ is the exact number of winning strategies for the existential player in $Q.F$. The size of the generated \qmice proof is linear in the number of nodes and edges of $G$.
\end{proof}

%See Appendix~\ref{qmice_sim_sec_ap} for the proof of the simulation. 
For an example of this simulation, see Table~\ref{qmice_example}, which shows a generated \qmice proof from the \dxyz of Figure~\ref{running_dag_fig} for the QBF $\Phi$ of Example~\ref{running_eg}.

\begin{table}[h]
\caption{Example of \qmice proof generated from a \cdxyz structure}
\centering
\addtolength{\tabcolsep}{6pt}
\begin{adjustbox}{width=0.85\textwidth}
\small
\begin{tabular}{c c c c}
\hline
Claim no. & Rule & Claim & \dxyz Edge no.\\ 
\hline
 $C_1$ & Axiom & $(\Phi|_{\{u_1=0, u_2=0\}},~\{e_1=1, e_2=0\},~ 1)$ & ${(9)}$\\ 
 $C_2$ & Absence of models & $(\Phi|_{\{u_1=0, u_2=0\}},~\{e_1=1, e_2=1\},~ 0)$ & ${(10)}$\\
 $C_3$ & Composition ($C_1, C_2$) & $(\Phi|_{\{u_1=0, u_2=0\}},~\{e_1=1\},~ 1)$ & ${(8)}$\\
 $C_4$ & $\simeq C_3$ & $(\Phi|_{\{u_1=0, u_2=1\}},~\{e_1=1\},~ 1)$ & ${(7)}$\\
 $C_5$ & Join ($C_3, C_4$) & $(\Phi|_{\{u_1=0\}},~\{e_1=1\},~ 1)$ & ${(4)}$\\
 $C_6$ & Absence of models & $(\Phi|_{\{u_1=0\}},~\{e_1=0\},~ 0)$ & ${(3)}$\\
 $C_7$ & Composition ($C_5, C_6$) & $(\Phi|_{\{u_1=0\}},~\emptyset,~ 1)$ & ${(1)}$\\
 $C_8$ & Axiom & $(\Phi|_{\{u_1=1\}},~\{e_1=1\},~ 2^2)$ & ${(6)}$\\
 $C_9$ & $\simeq C_5$ & $(\Phi|_{\{u_1=1\}},~\{e_1=0\},~ 1)$ & ${(5)}$\\
 $C_{10}$ & Composition ($C_8, C_9$) & $(\Phi|_{\{u_1=1\}},~\emptyset,~ 5)$ & ${(2)}$\\
 $C_{11}$ & Join ($C_7, C_{10}$) & $(\Phi,~\emptyset,~ 5)$ & -\\
 \hline
\end{tabular}
\end{adjustbox}
\label{qmice_example}
\end{table}

\begin{remark*}
As noted in~\cite{ChedeCKS26}, \qmice is capable of handling doubly exponential numbers using the arithmetic circuit in its lines. This helps \qmice to efficiently simulate \cdxyz proofs for QBFs with doubly exponential model counts.
\end{remark*}

As a Corollary to Theorem~\ref{thm:tqd_lowr_bound} and Theorem~\ref{thm:qmice_p_simulate_cert_tqdag}, we have the following:

\begin{corollary}
\qmice along with an equivalence rule is exponentially stronger than \cdxyz proof system.
\end{corollary}

\section{\xyz Compiler }\label{sec:tqt_compiler}
Based on the tree representation defined in Section~\ref{sec:xyz}, we implemented a C++ tool called \qbfsolver.
In Algorithm~\ref{alg:tqtcompiler}, we give a pseudo-code describing the implementation of \qbfsolver.
The algorithm proceeds in two broad parts: a recursive explorer to compile a \xyz structure for the given QBF, and a tree evaluation algorithm to compute the strategy count using the \xyz structure.
In what follows, we describe these two parts in detail.

First, we describe the recursive explorer presented in Algorithm~\ref{subr:recursive_explorer}, which explores the given formula in a breadth-first fashion.
The explorer starts with a root node, which corresponds to the input formula, and starts exploring according to the \xyz structure.
To make the recursion breadth-first, the explorer reads what to explore from a FIFO queue, and subsequent steps generated while exploring are added to the queue.
At each node, the explorer proceeds in five major steps listed below.

\begin{itemize}
    \item \textbf{Preprocessing}: The explorer receives a QBF \(Q. \phi\) along with a prefix assignment \(\alpha\).
    The explorer creates a new node with \(Q. \phi\) as the associated QBF and assigns it a unique id in line 1.
    In line 2, the explorer attempts to simplify the QBF by removing satisfied clauses and unused existential variables to obtain \(Q'. \phi'\) (see Algorithm~\ref{subr:simplify} in Appendix~\ref{solver_appx}).
    Note that removing satisfied clauses does not change the strategy count, but the strategy count may change when removing unused existential variables.
    However, these counts are related by the following expression, as described in~\cite{CapelliLPS24}.

    \begin{lemma}[{\cite[Lemma 2]{CapelliLPS24}}]
    \label{lem:var_elim}
    Consider a QBF \(Q. \phi\) and another \(Q'. \phi\) such that \(Q'\) contains all \(\forall\) variables from \(Q\) and a subset of the \(\exists\) variables in \(Q\).
    Then, their strategy counts are related as
    \[\#(Q. \phi) = 2^{\sum_{x \in Q \setminus Q'} 2^{\abs{\mathcal{L}_Q (x)}}} \#(Q'.\phi)\]
    \end{lemma}

    The explorer stores the exponent of \(2\) used for the multiplier in the expression, i.e. \(\sum_{x \in Q \setminus Q'} 2^{\abs{\mathcal{L}_Q (x)}}\), as an attribute of the node, denoted by \(e\).
    Storing the exponent, rather than the multiplier itself, preserves space required for the output.
    Moreover, we can keep \(e\) as a fixed-width integer, say \(64\) bits.
    This makes it possible to output the \xyz structure for larger QBF instances, such as, when \(\abs{\mathcal{L}_Q (x)}\) is large.
    For example, if \(\abs{\mathcal{L}_Q}\) is \(60\), \(e\) can accommodate \(2^{60}\), but \(2^{2^{60}}\) would be too large to store.

    \item \textbf{Caching}: 
    The explorer employs a caching mechanism to avoid exploring nodes it has already seen in lines 3-7, inspired from~\cite{CapelliLPS24}.
    In the cache, the explorer stores a node id associated to a hash key uniquely assigned to the QBF \(Q'. \phi'\).
    While exploring a node, say with node id \(N\), the explorer checks whether the key associated to the node is already present in the cache.
    If the key is already present in the cache, with node id \(M\) stored against it, the explorer marks \(N\) equivalent to \(M\) and marks \(N\) as SEEN.
    For this purpose, the explorer stores an attribute \(eq\) in node \(N\) and outputs \(N\) to terminate the current path.
If the cache does not contain the key, the explorer stores \(N\) against the node's key in line 8 and starts exploring.

    \item \textbf{Leaf Checks}:
    In line 9-10, if \(Q'. \phi'\) is an FQBF, then the explorer marks \(N\) as an FQBF\_LEAF and outputs \(N\) to terminate the path.
    To check whether \(Q'. \phi'\) is an FQBF, the explorer employs some syntactic checks and calls an FQBF oracle if needed.
    If this check fails, or the oracle call doesn't return a valid output, the explorer continues.

    In line 11-14, if \(\phi'\) is a tautology, then the explorer marks \(N\) as a TAUT\_LEAF.
    Then, it computes the count of strategies at this node as discussed in Section~\ref{prelims_sec}.
    Note that checking whether a CNF is a tautology is straightforward, as a CNF is a tautology if and only if it does not contain any clauses.
    Hence, we avoid any oracle calls in this step.
    The explorer outputs \(N\) and terminates the path.

    \item \textbf{Node Expansion}:
    If neither checks from the above pass, the explorer picks a variable \(a\) from the leftmost quantifier block in \(Q'\) in line 15.
    We denote the restricted QBF formulas after setting \(a\) as \(F_1 = Q'. \phi' \mid_{a=0}\) and \(F_2 = Q' \phi' \mid_{a=1}\).
    The explorer marks \(N\) as expanded in line 16.
    To continue exploration, the explorer adds a new node \(N_1\) with \((Q'.\phi' \mid_{a=0}, \alpha \cup \{a = 0\})\) as the input to continue the recursion in line 17.
    
    \item \textbf{Equivalence Check}:
    Now, the explorer checks whether \(F_1\) and \(F_2\) are equivalent as QBFs using an equivalence oracle call in line 18.
    If the formulas are equivalent, the explorer copies \(N_1\) as \(N_2\) in lines 19-20.
    Otherwise, the explorer adds a new node \(N_2\) with \((Q'.\phi' \mid_{a=1}, \alpha \cup \{a = 1\})\) as the input to continue the recursion in lines 21-22.
    Finally, in line 23, the explorer outputs an edge from \(N\) to \(N_1\) and \(N_2\) each.
\end{itemize}

The explorer outputs a file in a flat format to represent the tree.
The file contains two types of blocks of lines: node blocks and edge blocks.
A node block consists of six consecutive lines that contain the following information in the order listed below.

\begin{enumerate} 
\item Node id.
\item Depth in the recursion.
\item A value of \(0\) for internal nodes, and the strategy count at leaves.
\item The multiplier exponent \(e\) obtained from eliminating unused existential variables.
\item Id of the equivalent node \(eq\) in case of cache hits, and zero otherwise.
\item Label of the node, which can be LEAF\_FQBF, LEAF\_TAUT, SEEN or EXPANDED.
\end{enumerate}

Each edge block consists of one line which stores the parent node id, child node id, expanded variable ID and the value assigned to the variable.
This completes the description of the recursive explorer, the first part of \qbfsolver.
To illustrate the output, we describe the file output from the explorer for Example~\ref{running_eg}. Note that the \qbfsolver output structure is somewhere in between a \xyz and a \dxyz structure (see Figure~\ref{fig:solver_tree_eg}). Precisely, the \qbfsolver, implements syntactic checks (using hashing) for all nodes of the same level which is not used in \xyz. On the other hand, \dxyz uses the more powerful Skolem-equivalence checks for all nodes of the same level.

\begin{example}
\label{ex:solveroutput}
The explorer outputs a file with \(11\) nodes when run with Example~\ref{running_eg} as the input.
Node \(1\) expands on the variable \(u_1\) to get nodes \(2\) and \(3\) as the children with \(u_1 = 0\) and \(u_1 = 1\) respectively.
This node is output to the file as shown in Table~\ref{tab:outputnodes}.
The edges from \(1\) to \(2\) and \(3\) are output as follows.

{\sffamily
\begin{tabular}{lllll}
EDGE & 1 & 2 & 1 & 0 \\
EDGE & 1 & 3 & 1 & 1
\end{tabular}
}

Then, node \(2\) expands on \(e_1\) to get nodes \(4\) and \(5\) with \(e_1 = 0\) and \(e_1 = 1\) respectively.
Similarly, node \(3\) expands on \(e_1\) to get nodes \(6\) and \(7\) with \(e_1 = 0\) and \(e_1 = 1\) respectively.
Node \(4\) is found to be an FQBF, hence, marked as an FQBF\_LEAF (shown in Table~\ref{tab:outputnodes}).

Node \(5\) expands on \(u_2\), but finds that both \(u_2 = 0\) and \(u_2 = 1\) lead to the same restricted QBF.
Hence, the explorer creates a node \(8\) with \(u_2 = 0\), and outputs parallel edges from \(5\) to \(8\).
Now, at node \(6\), the restricted QBF after removing satisfied clauses becomes \(\forall u_2 \exists u_2 . (u_2 \vee \overline{e_2}) \wedge (\overline{u_2} \vee \overline{e_2})\).
Note that at node \(5\), it has the same QBF after removing satisfied clauses, hence node \(6\) is marked as SEEN.
The output for node \(6\) is shown in Table~\ref{tab:outputnodes}. 

At Node \(7\), the restricted formula after removing satisfied clauses is \(\forall u_2 \exists e_2 . \) with no clauses left.
Hence, it is a tautology and marked as LEAF\_TAUT.
Note that the propositional formula is empty, hence the explorer removes the existential variable \(e_2\) which depends on \(1\) universal variable.
Therefore, the exponent as per Lemma~\ref{lem:var_elim} after removing \(e_2\) is \(2^1 = 2\), leading to an output as shown in Table~\ref{tab:outputnodes}.

Then, node \(8\) expands on \(e_2\) to obtain nodes \(9\) and \(10\) with \(e_2 = 0\) and \(e_2 = 1\), which evaluate to \(1\) and \(0\) respectively.
The explorer has finished exploring.
Figure~\ref{fig:solver_tree_eg} shows the full tree as output by the explorer, wherein a double lined edge represent a cache hit.
\end{example}

\begin{table}
{\sffamily \small
\begin{tabular}{|ll | ll|}
\hline
NODE & 1 & NODE & 4\\
DEPTH & 0 & DEPTH & 2\\
VALUE & 0 & VALUE & 0 \\
MULTIPLIER\_EXPONENT & 0 & MULTIPLIER\_EXPONENT & 0 \\
EQUIVALENT\_NODE & 0 & EQUIVALENT\_NODE & 0 \\
STATUS & EXPANDED & STATUS & LEAF\_FQBF \\
\hline
NODE & 6 & NODE & 7 \\
DEPTH & 2 & DEPTH & 2 \\
VALUE & 0 & VALUE & 1 \\
MULTIPLIER\_EXPONENT & 0 & MULTIPLIER\_EXPONENT & 2 \\
EQUIVALENT\_NODE & 5 & EQUIVALENT\_NODE & 0 \\
STATUS & SEEN & STATUS & LEAF\_TAUT \\
\hline
\end{tabular}
}
\caption{Output format of some \xyz nodes from Example~\ref{ex:solveroutput}}
\label{tab:outputnodes}
\end{table}

The second part is the tree evaluation algorithm, which first builds the \xyz \(T\) in memory by reading the output of the explorer.
Then it starts a depth-first evaluator subroutine with memoization, with the root of \(T\) as the input, as described in Algorithm~\ref{subr:eval_tree}.
The subroutine is described in Algorithm~\ref{subr:eval_tree_helper}, and consists of three major steps.

\begin{itemize}
\item \textbf{Memo Check}:
The evaluator starts with searching the given node \(N\) in its memo, and returns the stored value if the node is found in the memo, in lines 1-3.

\item \textbf{Equivalence Check}:
Then, the evaluator first checks if the node is marked as SEEN in \(T\), using line 4.
If \(N\) is marked as SEEN, the evaluator makes a recursive call to obtain the value for \(N.eq\), in line 5.
Then, it multiplies the value by raising \(2\) with an exponent equal to the difference of the exponent stored at \(N\) and the exponent at \(N.eq\), in line 6.
This value is then stored in the memo in line 7.

\item \textbf{Leaf Check}:
If the node is marked as LEAF\_FQBF, the evaluator stores \(0\) in the memo against \(N\), in lines 8-9.
Otherwise, if the node is marked as LEAF\_TAUT, the evaluator multiplies the value \(N.v\) stored at the node, multiplies it with \(2^{N.e}\) and stores it in the memo against \(N\), in lines 10-11.

\item \textbf{Evaluating Children}:
For internal nodes, the evaluator makes recursive calls on the children of \(N\) to obtain the children's value in lines 12-14.
Then, if \(N\) expands on a universal variable, the subroutine multiplies the children's values, in lines 15-16.
Otherwise, if \(N\) expands on a existential variables, the subroutine adds the values, in lines 17-19.
The obtained value is multiplied by \(2^{N.e}\) and stores it in the memo against \(N\), in line 20.

Finally, the value stored in the memo against \(N\) is returned from the evaluator.
\end{itemize}
The tree evaluation algorithm outputs the value obtained from the evaluator subroutine for the root of \(T\).
This completes the description of the \qbfsolver.

\noindent
\begin{minipage}{0.45\textwidth}
\begin{algorithm}[H]
\footnotesize
\begin{algorithmic}
\Require A QBF \(Q. \phi\).
\Ensure A \xyz \((Q, T)\) and \#QBF for \(Q. \phi\).
\State \((Q, T) \leftarrow \mathop{Explorer}(Q. \phi, \alpha=\emptyset)\)
\State Output \(((Q, T), \mathop{EvaluateTree}((Q, T)))\)
\end{algorithmic}
\caption{\#QBF Counter}
\label{alg:tqtcompiler}
\end{algorithm}
\end{minipage}
\hfill
\begin{minipage}{0.52\textwidth}
\begin{algorithm}[H]
\footnotesize
    \begin{algorithmic}
        \Require A \xyz structure \((Q, T)\).
        \Ensure An integer \(c\) with the value at the root of T.
        \State \(\mathsf{memo} \leftarrow \{\}\)
        \State Output \(\mathop{EvaluateTreeHelper}((Q, T), \mathop{root}(T))\)
    \end{algorithmic}
    \caption{$\mathop{EvaluateTree}$}
    \label{subr:eval_tree}
\end{algorithm}
\end{minipage}

\noindent
\begin{minipage}{0.46\textwidth}
\begin{algorithm}[H]
\footnotesize
\begin{algorithmic}[1]
\Require A QBF \(Q. \phi\), a partial assignment \(\alpha\).
\Ensure A \xyz structure \((Q, T)\).
\State \(N \leftarrow \mathsf{Node}(F=Q.\phi, e=0, v=0)\)
\State \(Q'. \phi', N.e \leftarrow \textit{Simplify}(Q. \phi)\)
\State \(k \leftarrow \mathop{ComputeHash}(Q'. \phi')\)
\If {\(\mathop{Cache}(k)\)}
\State \(N.eq \leftarrow \mathop{Cache}(k)\)
\State Output \(N\) as SEEN
\EndIf
\State \(\mathop{Cache}(k) \leftarrow N\)
\If {\(\mathop{FQBF}(Q'. \phi')\)}
\State Output \(N\) as FQBF\_LEAF
\ElsIf {\(\mathop{TAUT}(\phi)\)}
\State \(N.v \leftarrow \mathop{TautCountStrategies}(Q'. \phi')\)
\State Output \(N\) as TAUT\_LEAF
\EndIf
\State Denote the leftmost variable in \(Q'. \phi'\) as \(a\)
\State Output \(N\) as EXPANDED
\State \(N_1 \leftarrow \mathop{Explorer}(Q'. \phi' \mid_{a = 0}, \alpha \cup \{a = 0\})\)
\If {\(\mathop{Equiv}(Q'. \phi' \mid_{a = 0}, Q'. \phi' \mid_{a=1})\)}
\State \(N_2 \leftarrow N_1\)
\Else
\State \(N_2 \leftarrow \mathop{Explorer}(Q'. \phi' \mid_{a=1}, \alpha \cup \{a = 1\})\)
\EndIf
\State Output an edge from \(N\) to \(N_1\) and \(N_2\) each
\end{algorithmic}
\caption{$\mathop{Explorer}$}
\label{subr:recursive_explorer}

\end{algorithm}
\end{minipage}
\hfill
\begin{minipage}{0.51\textwidth}
\begin{algorithm}[H]
\footnotesize
    \begin{algorithmic}[1]
        \Require A \xyz \((Q, T)\) and a node \(N\) in \(T\).
        \Ensure An integer \(c\) with the value calculated at \(N\)
        \If {\(\mathsf{memo}(N)\)}
        \State Return \(\mathsf{memo}(N)\)
        \EndIf
        \If {\(N.eq\) is set}
        \State \(c \leftarrow \mathop{EvaluateTree}((Q, T), N.eq)\)
        \State \(c \leftarrow 2^{N.e - N.eq.e} \times c\)
        \State \(\mathsf{memo}(N) \leftarrow c\)
        \ElsIf {\(N\) is an FQBF leaf}
        \State \(\mathsf{memo}(N) \leftarrow 0\)
        \ElsIf {\(N\) is a TAUT leaf}
        \State \(\mathsf{memo}(N) \leftarrow 2^{N.e} \times N.v\)
        \Else
        \State \(c_1 \leftarrow \mathop{EvaluateTreeHelper}((Q, T), \mathop{left}(N))\)
        \State \(c_2 \leftarrow \mathop{EvaluateTreeHelper}((Q, T), \mathop{right}(N))\)
        \If{\(N\) is universal node}
        \State \(c \leftarrow c_1 \times c_2\)
        \Else
        \State \(c \leftarrow c_1 + c_2\)
        \EndIf
        \State \(\mathsf{memo}(N) \leftarrow 2^{N.e} \times c\)
        \EndIf
        \State Return \(\mathsf{memo}(N)\)
    \end{algorithmic}
    \caption{$\mathop{EvaluateTreeHelper}$}
    \label{subr:eval_tree_helper}
\end{algorithm}
\end{minipage}

\subsection{Experimental evaluation}\label{subsec:exp_eva}
We implement the algorithms discussed above as a C++ tool named \qbfsolver. 
We also provide two new families of QBF benchmarks. 
The experimental evaluation is performed on an Amazon Web Services (AWS) EC2 `r7a.2xlarge' instance running Ubuntu. All benchmarks were executed on a single core under identical conditions with a 900-second timeout and a 32 GB memory limit per instance. We report wall-clock runtimes throughout the experiments. An instance was considered unsolved if the solver exceeded either the prescribed time limit or the available memory. We have used two existing solvers as Oracles in the implementation of \qbfsolver. Specifically, the DepQBF Solver~\cite{LonsingE17} to check for False QBFs and the QSOLE equivalence checker~\cite{PfeifferPGS26, PfeifferGS25} to find QBF equivalences between the two children of a node. By default, \qbfsolver uses both these oracles. To be precise, it calls QSOLE in line-18 and DepQBF in line-9 of Algorithm~\ref{subr:recursive_explorer}. For the experiments reported in this paper, we disable \qbfsolver from calling QSOLE by passing an optional parameter ``- -no\_qsole''. However \qbfsolver does some other basic equivalence checks in line-18 of Algorithm~\ref{subr:recursive_explorer}. In future we hope to find benchmarks which depict the use of QSOLE in \qbfsolver.

\textbf{Benchmarks:} We evaluate our solver in comparison to d4-QBF~\cite{CapelliLPS24} on benchmarks from~\cite{CapelliLPS24} and also introduce new benchmarks to capture the uniqueness of \qbfsolver. We introduce the following two benchmarks:
\begin{itemize}
    \item Rec\_Alt\_Taut\_False formulas: These are the recursively alternating tautology and FQBF formulas from Definition~\ref{ctqt_ub_formula}. For an integer $n>1$, the corresponding QBF has $2n-1$ variables and $n$ clauses. It has $2^{2^n-n-1}$ winning strategies for existential player. 

    \item Phi\_True formulas: These are the true versions of important $\phi_n$ formulas in the QBF literature~\cite{DBLP:journals/tcs/JanotaM15} (see Definition~\ref{phi_dag_ub_formula}). For an integer $n>1$, the corresponding QBF has $4n$ variables and $4n$ clauses of width $2$. These formulas have doubly exponential many winning strategies, to be precise, $2^{2^{n+1}-2}$. 
\end{itemize}

\begin{figure}
\centering
\includegraphics[width=0.7\textwidth]{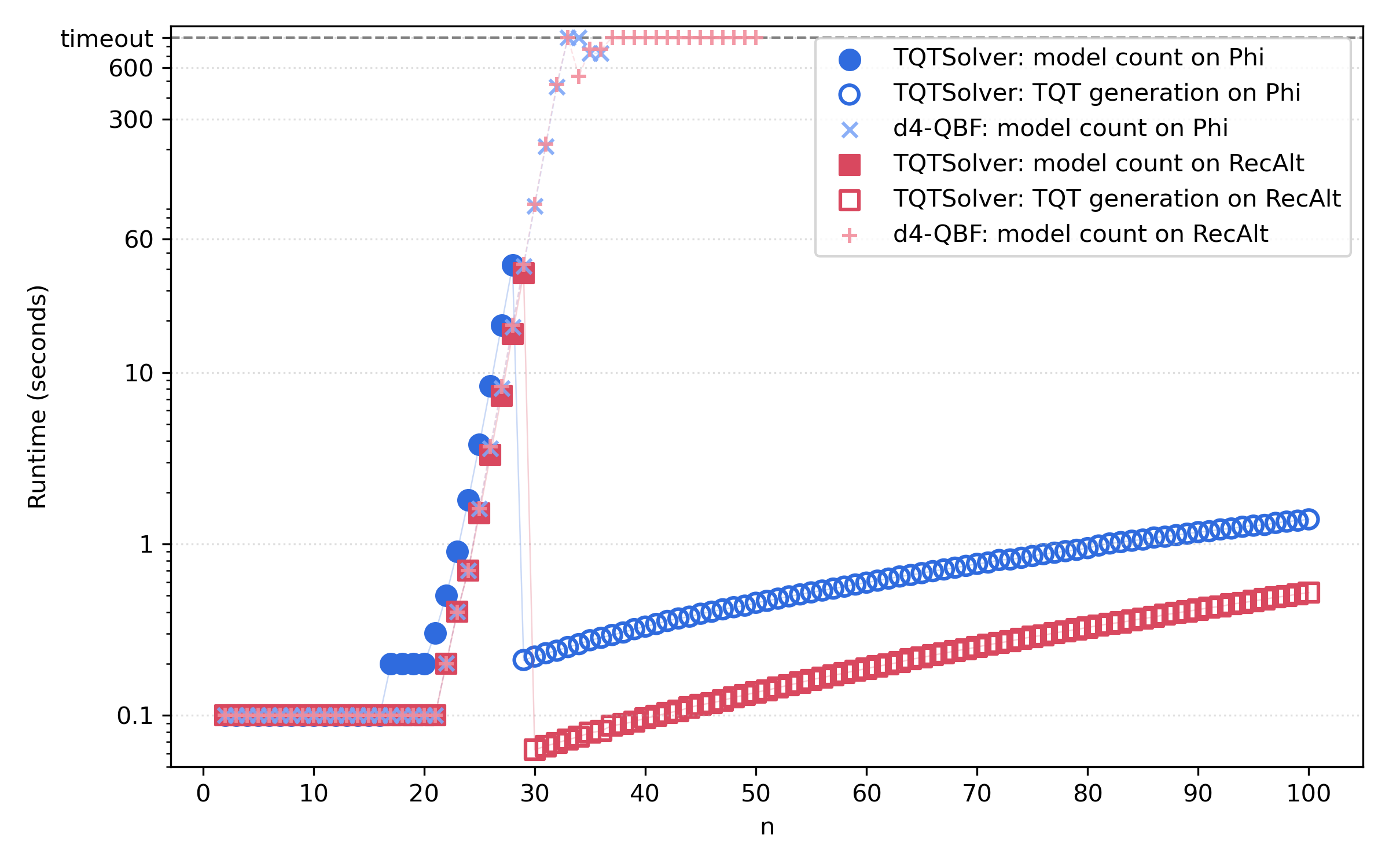}
    \caption{d4-QBF and \qbfsolver on new benchmarks}
    \label{graph2_solver}
\end{figure}

Any solver trying to output exact counts is bound to take exponential time computing these doubly exponential values, which is also the case with d4-QBF. However as shown in Figure~\ref{graph2_solver}, after n=32, d4-QBF keeps running out of memory or time but \qbfsolver can output the equivalent small-sized \xyz structure in under 5 seconds. Figure~\ref{table_solver_results} shows a comparison considering the number of solved instances for all the benchmarks families considered for evaluations.

Explanation for Figure~\ref{graph2_solver}: We observe that for the new benchmarks, both d4-QBF and \qbfsolver are computing the count correctly in under 10 seconds up to $n=28$. Note that the model count at this point is $2^{2^{30}}$. In fact d4-QBF can output the model counts up to $n=33$. After that both solvers are timed-out and unable to compute the count. 
However, we can run \qbfsolver just to output the \xyz representation without computing the number of models. Running this on the two new benchmarks it outputs the \xyz representations for up to $n=100$ in under $5$ seconds. 

\begin{figure}
    \centering
\includegraphics[width=0.55\textwidth]{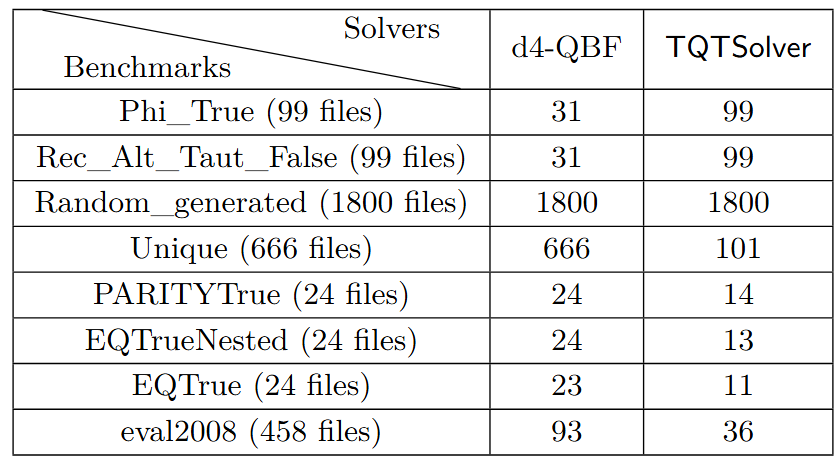}
    \caption{Number of solved instances from benchmarks}
\label{table_solver_results}
\end{figure}

Explanation for Figure~\ref{table_solver_results}: These compare the performance of d4-QBF against that of \qbfsolver on eight different benchmark families and reports the number of solved instances by each solver. 
The total number of instances in each family are mentioned in the parenthesis. As already seen in Figure~\ref{graph2_solver}, \qbfsolver can produce a \xyz for all instances in the new benchmark families while the d4-QBF can solve 31 smallest instances in these benchmarks within the time limit of 900 seconds. For the remaining benchmarks in this table, we enable the model counting in \qbfsolver as well.
For the Random generated benchmarks, both solvers can solve all the instances. In the Unique benchmarks (where the model count is always $1$), d4-QBF can solve all of the instances whereas \qbfsolver could only solve $15\%$ of them. Parity is hard in \xyz (Theorem~\ref{thm:tqt_lowr_bound}), so \qbfsolver expectedly only solves the smallest 14 instances from the PARITYTrue benchmarks with model count $1$ while d4-QBF solves all $24$ of them. However after the first 18 instances, the time taken by the d4-QBF is observed to be doubling to that of the previous instance. Hence, extending the PARITYTrue beyond $26$ instances would imply a timeout for d4-QBF. Similarly for the benchmark families based on Equality (EQTrue, EQTrueNested) the model count is doubly exponential and d4-QBF can solve $23$ or $24$ of the smallest instances but extending these families beyond $25$ would imply a timeout for d4-QBF. On the other hand, \qbfsolver on these benchmarks can solve up-to $13$ of the smallest instances. Lastly on the QBF solving benchmarks from 2008, both solvers perform poorly.

Runtime comparisons for these benchmarks on both solvers are shown in  Figure~\ref{all_benchmark_times}. Note that not all benchmark instances are plotted in Figure~\ref{all_benchmark_times}. To reduce the clutter, we randomly sample around $15\%$ of instances from the bigger benchmark families (Random, Unique) and around $20\%$ of the timeout instances on either side.

\begin{figure}[!ht]
    \centering
\includegraphics[width=0.8\linewidth]{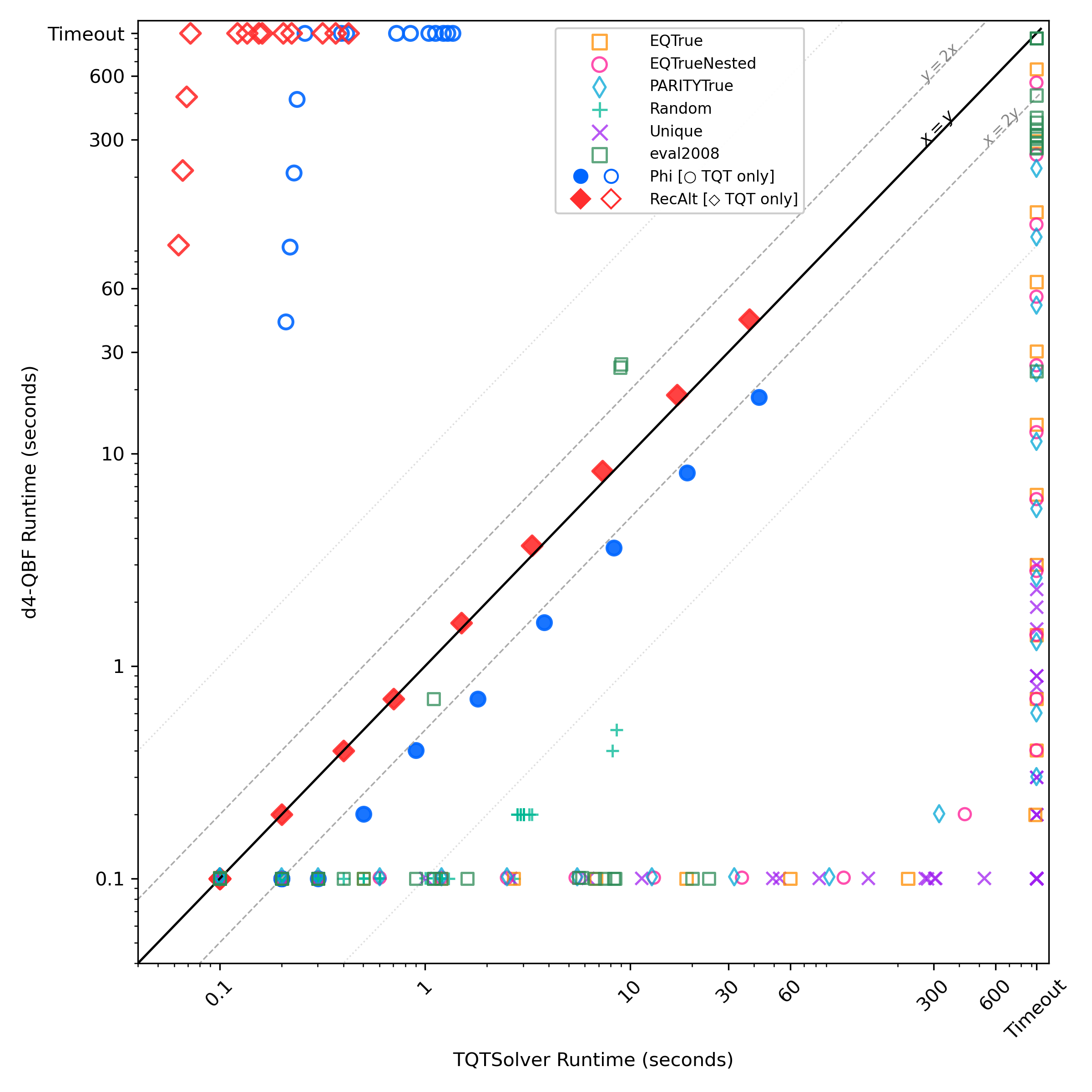}
    \caption{d4-QBF vs \qbfsolver on all benchmarks}
\label{all_benchmark_times}
\end{figure}
%\vspace{1cm}

\section{Conclusion and Future Works}
This paper explores the knowledge compilation approach to solve intractable problems for QBFs, including the \#QBF problem. The major contributions are the introduction of new target languages for QBFs which supports important QBF queries efficiently. In addition to the theoretical analysis of the strength and limitations of the proposed languages, the paper also develops a compiler \qbfsolver to generate the simplest \xyz representation for a given QBF. We hope to improve our solver by incorporating more engineering and optimization techniques useful in existing SAT, \#SAT, and \#QBF solvers. 
An immediate direction is to update the \qbfsolver to also answer the MaxQBF problem.

Observe that the \cxyz (\cdxyz) proof systems are also sound and complete for \#SAT (propositional model counting). It is open to relate the proposed proof systems with the existing \#SAT proof systems \cpog \cite{bryant2023certified} and \kcps~\cite{Capelli19}. In fact, a lower-bound for \cpog is open in the literature. Moreover, in \cxyz, existential (universal) nodes can be viewed as OR (AND), but these AND-gates are not decomposable as required in \cpog. For propositional unsatisfiable formulas, the size of \kcps proofs is equivalent to a restricted version of Resolution~\cite[Proposition 5.1]{BeyersdorffFHHK26} where lower bounds are known~\cite{Haken85,BeameP96}. Such formulas have easy \cxyz proofs when a strong FQBF proof system is used.

Mengel and Slivovsky~\cite{DBLP:conf/sat/MengelS21} formalize the QBDD system~\cite{DBLP:conf/cp/PanV04} as a QBF proof system. They extended propositional OBDD (ordered binary decision diagrams) proof systems~\cite{DBLP:conf/coco/BussIKS18} to QBFs by adding the $\forall$-reduction rule~\cite{DBLP:journals/jacm/BeyersdorffBCP20}. Note that the QBDD system is a dynamic QBF-proof-system whereas cert-TQD is a static \#QBF-proof-system. The proof-theoretic comparison of these remains open.

While \dxyz structure is built starting from the outermost quantifier, Capelli and Mengel~\cite{DBLP:conf/stacs/CapelliM19} use another approach which starts from the propositional formula and gradually builds sd-DNNFs (structured deterministic DNNF) recursively from inner to outermost quantifier. It would be interesting to compare \dxyz structure to this system.

\bibliography{myref}

\newpage
\appendix
\section*{Appendix}

\section{Tractable Queries on \fxyz}\label{sec:query_support_tqt}

The list of tractable queries in \xyz are described below.
\begin{itemize}
    \item \textbf{\#\xyz:} Given a \xyz structure $G=(Q,T)$, \#\xyz$(G)$ is the sum of weights of distinct model trees in $T$.

This has a pseudo-polynomial algorithm in \xyz as follows: maintain a variable $val$ at every node of $T$. In the bottom-up fashion, compute the value of $val$ at every node. At the leaves whose parent node is labelled by $x_n$, $val=$ node label itself. At the leaves whose label is `$0$F', the $val=0$ and at the leaves whose parent is not labelled by $x_n$, $val$ is equal to $2^{\text{node-label}}$.
At intermediate nodes, if the label is an $\exists$-variable, $val$ is the sum of $val$s of children nodes and if the label is an $\forall$-variable, $val$ is the product of $val$s of children nodes. Note that in case of multi-edges, we consider the single child node twice in this computation. Finally at the root we have the \#\xyz$(G)$ answer.

\item \textbf{Conjunction of two \xyzs with the same quantifier sequence:} Given two \xyz structures $G_1=(Q,T_1)$ and $G_2=(Q,T_2)$ corresponding to some quantified statements $Q.\phi_1$ and $Q.\phi_2$, $G_1\land G_2$ is the equivalent \xyz structure $G_3=(Q,T_3)$ such that $G_3$ corresponds to the quantified statement $Q.\phi_1\land \phi_2$.

This is tractable in \xyz as follows: We recursively generate the \xyz $T_3$ in the top-down fashion (starting from the root) from $T_1,T_2$. If corresponding nodes in both have the same labels, create a similar node in $T_3$ as well. If corresponding nodes in both are leaves whose parents are labelled by the innermost variable of $Q$, create a leaf for $T_3$ with value after conjunction of both leaves. If one of the nodes is a FQBF, irrespective of the other tree, make a FQBF leaf node in $T_3$ as well. If one of the nodes is a TAUT, create the remaining nodes according to the other \xyz for the remaining of this path. If one of the nodes has a outgoing multi-edge but the other doesn't, then consider the single child twice when comparing corresponding nodes in the construction. The size of $T_3$ formed this way would be $O(|T_1|+|T_2|)$.

\item \textbf{Ordered Conjunction of two \xyzs with disjoint variables:} Given two \xyz structures (in this order) $G_1=(Q_1,T_1)$ and $G_2=(Q_2,T_2)$, corresponding to some quantified statements $Q_1.\phi_1$ and $Q_2.\phi_2$ where $\vars(Q_1)\cap\vars(Q_2)=\emptyset$. 
$G_1\land G_2$ is the equivalent \xyz structure $G_3=(Q_1Q_2,T_3)$ such that $G_3$ corresponds to the quantified statement $Q_1Q_2.\phi_1\land \phi_2$.

This is tractable in \xyz as follows: We recursively generate the \xyz $T_3$ in the top-down fashion (starting from the root) from $T_1,T_2$. Make a copy of $T_1$ first to start constructing $T_3$. Now at the $0$-leaves of $T_1$, no need to add any nodes from $T_2$. At the $1$-leaves of $T_1$, add the entire $T_2$ tree to replace the leaf node. At the TAUT-nodes of $T_1$ at level $i$, create nodes $x_i,...,x_n\in Q_1$ and add multi edges for consecutive nodes in this sequence. Finally after $x_n$, add a multi edge to the start of the entire $T_2$ tree. The size of $T_3$ formed this way would be $O(|\vars(Q_1)|\cdot|T_1|\cdot|T_2|)$.

\item \textbf{Conditioning of $\exists$-variables in \xyz:} 
Given a \xyz structure $G=(Q,T)$ corresponding to some quantified statement $Q.\phi$ and a (partial) assignment $\alpha$ to $\vars_\exists(Q)$, $G|_\alpha$ is the \xyz structure $G':=(Q,T')$ such that $G'$ corresponds to the quantified structure $(Q.F)|_\alpha$.

This is tractable in \xyz as follows: We create $T'$ from $T$ in a top-down fashion in the quantifier sequence of variables of $\alpha$. Say $x_i$ is the first variable in $\alpha$, in the tree $T$, we cut the outgoing child labelled with $\lnot\alpha(x_i)$ from all the $x_i$ nodes and insert a FQBF node in its place. Note that the subtree from the outgoing child labelled with $\alpha(x_i)$ is still in the tree as is. Recursively we do this for all variables of $\alpha$.
\end{itemize}

\section{Missing \qbfsolver subroutines and evaluation graphs}\label{solver_appx}

\begin{algorithm}
\footnotesize
\begin{algorithmic}
\Require A QBF \(Q. \phi\) with a partial assignment \(\alpha\).
\Ensure A reduced QBF \(Q'. \phi'\), and an integer \(e\).
\State \(Q. \phi' \leftarrow \mathop{ReduceClauses}(Q. \phi)\)
\State \(Q'. \phi', e \leftarrow \mathop{RemoveUnusedExistentials}(Q. \phi')\)
\State Return \((Q'. \phi', e)\)
\end{algorithmic}
\caption{$\textit{Simplify}$}
\label{subr:simplify}
\end{algorithm}

\noindent
\begin{minipage}{0.475\textwidth}    
\begin{algorithm}[H]
\footnotesize
\begin{algorithmic}[1]
\Require A QBF \(Q. \phi\), a partial assignment \(\alpha\).
\Ensure A reduced QBF \(Q. \phi'\).
\State \(\phi' \leftarrow []\)
\For{Each clause \(C\) in \(\phi\)}
\State \(C' \leftarrow [], sat \leftarrow False\)
\For{Each literal \(l\) in \(C\)}
\If{\(var(l)\) is assigned in \(\alpha\)}
\If{\(l\) evaluates to true in \(\alpha\)}
\State
\(sat \leftarrow true\)
\State
\textbf{Break}
\EndIf
\Else
\State Append \(l\) to \(C'\)
\EndIf
\EndFor
\If{Not \(sat\)}
\State Append \(C'\) to \(\phi'\)
\EndIf
\EndFor
\State Return \(Q. \phi'\)
\end{algorithmic}
    \caption{$\mathop{ReduceClauses}$}
    \label{subr:reduce_clauses}
\end{algorithm}
\end{minipage}
\hfill
\begin{minipage}{0.475\textwidth}
\begin{algorithm}[H]
\footnotesize
\begin{algorithmic}[1]
\Require A QBF \(Q. \phi\).
\Ensure A reduced QBF \(Q'. \phi\), and an integer \(e\).
\State Initialize \(Q' \leftarrow [], e \leftarrow 0, p \leftarrow 0\)
\For{Each quantifier block \(B\) in \(Q\)}
\If{\(B\) is a universal block}
\State \(Q' \leftarrow Q' \cup B\)
\State \(p \leftarrow p + \lvert B \rvert\)
\State \textbf{Continue}
\EndIf
\State Initialize \(B' \leftarrow []\)
\For{Each variable \(v\) in \(B\)}
\If{\(v\) appears in a clause in \(\phi\)}
\State
\(B' \leftarrow B' \cup \{v\}\)
\Else
\State \(e \leftarrow e + 2^p\)
\EndIf
\EndFor
\If{\(B' \neq \emptyset\)}
\State \(Q' \leftarrow Q' \cup B'\)
\EndIf
\EndFor
\State Return \((Q'. \phi, e)\)
\end{algorithmic}
\caption{$\mathop{RemoveUnusedExistentials}$}
    \label{subr:unused_vars}
\end{algorithm}
\end{minipage}

\end{document}